\documentclass[12pt]{article}
\usepackage{amsmath,amsthm,enumerate}
\usepackage{amssymb}

\usepackage{float} 
\usepackage[hyperindex,colorlinks,citecolor=blue,pdftex,debug]{hyperref}
\usepackage{bookmark}
\usepackage[all]{hypcap} 
\usepackage[algo2e,algosection,tworuled,noend,noline]{algorithm2e}
\usepackage{graphicx}

\newcommand{\myTimes}{
  \usepackage{txfonts}
  \renewcommand{\sfdefault}{cmss}
  \DeclareMathAlphabet{\mathsf}{OT1}{cmss}{m}{n}
  \SetMathAlphabet{\mathsf}{bold}{OT1}{cmss}{b}{n}
  \renewcommand{\ttdefault}{cmtt}
  \DeclareMathAlphabet{\mathtt}{OT1}{cmtt}{m}{n}
  \SetMathAlphabet{\mathtt}{bold}{OT1}{cmtt}{b}{n}

\DeclareSymbolFont{cmsymbols}     {OMS}{cmsy}{m}{n}
\DeclareSymbolFont{ztmcmsymbols}     {OMS}{ztmcm}{m}{n}
\DeclareRobustCommand*{\cmmathcal}[1]{\gdef\F@ntPrefix{m@thcalch@r}%
  \@EachCharacter ##1\@EndEachCharacter}

\DeclareMathSymbol{\m@thcalch@rA}{\mathord}{cmsymbols}{65}
\DeclareMathSymbol{\m@thcalch@rB}{\mathord}{cmsymbols}{66}
\DeclareMathSymbol{\m@thcalch@rC}{\mathord}{cmsymbols}{67}
\DeclareMathSymbol{\m@thcalch@rD}{\mathord}{cmsymbols}{68}
\DeclareMathSymbol{\m@thcalch@rE}{\mathord}{cmsymbols}{69}
\DeclareMathSymbol{\m@thcalch@rF}{\mathord}{cmsymbols}{70}
\DeclareMathSymbol{\m@thcalch@rG}{\mathord}{cmsymbols}{71}
\DeclareMathSymbol{\m@thcalch@rH}{\mathord}{cmsymbols}{72}
\DeclareMathSymbol{\m@thcalch@rI}{\mathord}{cmsymbols}{73}
\DeclareMathSymbol{\m@thcalch@rJ}{\mathord}{cmsymbols}{74}
\DeclareMathSymbol{\m@thcalch@rK}{\mathord}{cmsymbols}{75}
\DeclareMathSymbol{\m@thcalch@rL}{\mathord}{cmsymbols}{76}
\DeclareMathSymbol{\m@thcalch@rM}{\mathord}{cmsymbols}{77}
\DeclareMathSymbol{\m@thcalch@rN}{\mathord}{cmsymbols}{78}
\DeclareMathSymbol{\m@thcalch@rO}{\mathord}{cmsymbols}{79}
\DeclareMathSymbol{\m@thcalch@rP}{\mathord}{cmsymbols}{80}
\DeclareMathSymbol{\m@thcalch@rQ}{\mathord}{cmsymbols}{81}
\DeclareMathSymbol{\m@thcalch@rR}{\mathord}{cmsymbols}{82}
\DeclareMathSymbol{\m@thcalch@rS}{\mathord}{cmsymbols}{83}
\DeclareMathSymbol{\m@thcalch@rT}{\mathord}{cmsymbols}{84}
\DeclareMathSymbol{\m@thcalch@rU}{\mathord}{cmsymbols}{85}
\DeclareMathSymbol{\m@thcalch@rV}{\mathord}{cmsymbols}{86}
\DeclareMathSymbol{\m@thcalch@rW}{\mathord}{cmsymbols}{87}
\DeclareMathSymbol{\m@thcalch@rX}{\mathord}{cmsymbols}{88}
\DeclareMathSymbol{\m@thcalch@rY}{\mathord}{cmsymbols}{89}
\DeclareMathSymbol{\m@thcalch@rZ}{\mathord}{cmsymbols}{90}

\let\mathcal\cmmathcal
}

\usepackage{cheatpf}
\def\kwfont{\itshape} 
\pfshortnumbers{6} 
\SetAlgoInsideSkip{smallskip}
\makeatletter
\renewcommand{\algocf@caption@tworuled}{\box\algocf@capbox\kern-\interspacetitleruled}
\makeatother
\newenvironment{algorule}[1][htb]{\begin{algorithm2e}[#1]\SetAlFnt{\small\sf}\advance\lineskip by
2pt}{\end{algorithm2e}}
\newcommand{\mykeyword}[1]{\text{\KwSty{#1} }}
\newcommand{\algOr}{\mykeyword{or}}
\newcommand{\algAnd}{\mykeyword{and}}
\newcommand{\algNot}{\mykeyword{not}}

\DontPrintSemicolon 
\SetKwIF{If}{ElseIf}{Else}{if}{then}{else if}{else}{}

\SetCommentSty{textrm}
\SetKwComment{Comment}{/\hspace{-0.2em}/ }{}
\SetArgSty{textrm}
\SetNlSty{textrm}{(}{)}

\newcommand{\shownotes}{1}
\ifnum\shownotes=1
\newcommand{\authnote}[3]{{ \textcolor{#3}{$\langle\hspace{-0.2em}\langle$\textsf{\footnotesize #1: #2}$\rangle\hspace{-0.2em}\rangle$}}}
\else
\newcommand{\authnote}[2]{}
\fi

\newcommand{\sign}{{\operatorfont sign}}

\newcommand{\cei}[1]{{\lceil #1\rceil}}

\makeatletter
\def\@@enum@[#1]{%
  \@enLab{}\let\@enThe\@enQmark
  \@enloop#1\@enum@
  \expandafter\edef\csname label\@enumctr\endcsname{\the\@enLab}%
  \expandafter\let\csname the\@enumctr\endcsname\@enThe
  \csname c@\@enumctr\endcsname7
  \expandafter\settowidth
            \csname leftmargin\romannumeral\@enumdepth\endcsname
            {\the\@enLab\hspace{\labelsep}}%
  \@enum@}
\makeatother

\newenvironment{bullets}{\leftmargini=1em\begin{itemize}}{\end{itemize}}

\newenvironment{flushdescription}{\leftmargini=1em\begin{description}}
{\end{description}}

\newcommand\ang[1]{{\mathopen\langle #1\mathclose\rangle}}
\newcommand{\df}[1]{\emph{#1}}
\newcommand\set[1]{\mathopen\{#1\mathclose\}}

\newcommand\Paren[1]{{\left( #1\right)}}

\newcommand{\lint}[2]{[#1,#2)}

\newcommand{\fld}[1]{\ensuremath{\textit{#1}}}
\newcommand{\rul}[1]{\ensuremath{\texttt{\slshape #1\/}}}

\newcommand{\increment}[1]{#1\mathord{+}\mathord{+}}
\newcommand{\decrement}[1]{#1\mathord{-}\mathord{-}}

\newcommand{\encode}{\mathrm{encode}}
\newcommand{\decode}{\mathrm{decode}}

\newcommand{\D}{D}
\newcommand{\f}{f}
\newcommand{\front}{\mathrm{front}}
\newcommand{\R}{R}

\newcommand{\Alarm}{\rul{Alarm}}
\newcommand{\Comp}{\rul{Compute}}
\newcommand{\Calc}{\rul{Calculate}}
\newcommand{\Transfer}{\rul{Transfer}}
\newcommand{\MoveBase}{\rul{MoveBase}}
\newcommand{\MAJ}{\rul{Interval-plur}}
\newcommand{\Move}{\rul{Move}}

\newcommand{\ruSwing}{\rul{Swing}}

\newcommand{\maj}{\mathop{\mathrm{maj}}}

\newcommand{\Addr}{\fld{Addr}}
\newcommand{\cAddr}{\fld{c.Addr}}
\newcommand{\Rec}{\fld{Rec}}
\newcommand{\cRec}{\fld{c.Rec}}
\newcommand{\Info}{\fld{Info}}
\newcommand{\cInfo}{\fld{c.Info}}

\newcommand{\Drift}{\fld{Drift}}
\newcommand{\cDrift}{\fld{c.Drift}}

\newcommand{\Zigzag}{\rul{Zigzag}}
\newcommand{\ZigDepth}{\fld{ZigDepth}}
\newcommand{\ZigDir}{\fld{ZigDir}}

\newcommand{\cHold}{\fld{c.Hold}}
\newcommand{\Mode}{\fld{Mode}}
\newcommand{\Stage}{\fld{Stage}}

\newcommand{\Normal}{\mathrm{Normal}}

\newcommand{\Marking}{\mathrm{Marking}}
\newcommand{\Planning}{\mathrm{Planning}}

\newcommand{\Mopping}{\mathrm{Mopping}}
\newcommand{\Recovering}{\mathrm{Recovering}}

\newcommand{\Mark}{\rul{Mark}}

\newcommand{\Plan}{\rul{Plan}}
\newcommand{\RangeCheck}{\rul{RangeCheck}}

\newcommand{\Mop}{\rul{Mop}}

\newcommand{\State}{\fld{State}}

\newcommand{\Core}{\fld{Core}}
\newcommand{\cCore}{\fld{c.Core}}

\newcommand{\cState}{\fld{c.State}}
\newcommand{\Sweep}{\fld{Sw}}
\newcommand{\cSweep}{\fld{c.Sw}}

\newcommand{\Coordinated}{\mathrm{Coordinated}}
\newcommand{\Dir}{\fld{Dir}}
\newcommand{\dir}{\mathrm{dir}}

\newcommand{\bbZ}{\mathbb{Z}}

\newcommand{\customqed}[1]{{\renewcommand{\qedsymbol}{#1}\qed}}
\newcommand{\varqed}{\customqed{\hbox{$\lrcorner$}}}

\numberwithin{equation}{section}

\theoremstyle{plain}

\newtheorem{theorem}{Theorem}[section]
\newtheorem{lemma}[theorem]{Lemma}

\theoremstyle{definition}
\newtheorem{Definition}[theorem]{Definition}
\newenvironment{definition}{\begin{Definition}}{\varqed\end{Definition}}
\newtheorem{Example}[theorem]{Example}
\newenvironment{example}{%
 \begin{Example}}{\varqed\end{Example}}

\newtheorem{Problem}[theorem]{Problem}

\theoremstyle{remark}
\newtheorem{Remark}[theorem]{Remark}
\newenvironment{remark}{\begin{Remark}}{\varqed\end{Remark}}

\newcommand{\txt}[1]{\text{\rmfamily\mdseries\upshape{#1}}}

\newcommand{\cA}{\mathcal{A}}

\newcommand{\cI}{\mathcal{I}}

\newcommand{\cS}{\mathcal{S}}

\newcommand{\K}{K}

\newcommand{\blank}{\text{\textvisiblespace}}
\newcommand{\state}{\fld{state}}
\newcommand{\start}{\fld{start}}
\newcommand{\tape}{\fld{tape}}
\newcommand{\pos}{\fld{pos}}
\newcommand{\h}{h}

\newcommand{\Last}{\mathrm{Last}}
\newcommand{\TransferLast}{\mathrm{TransferLast}}
\newcommand{\Configs}{\mathrm{Configs}}
\newcommand{\E}{E}
\newcommand{\Histories}{\mathrm{Histories}}

\newcommand{\Z}{Z}

\newcommand{\TransferStart}{\mathrm{TfSt}}
\newcommand{\TransferSw}{\mathrm{TfSw}}

\begin{document}

\title{A Turing Machine Resisting Isolated Bursts Of Faults}
\author{
Ilir \c{C}apuni  \hspace{1cm} Peter G\'{a}cs  \\
Boston University \\
Department Of Computer Science \\
111 Cummington str, Boston MA \\
\{ilir, gacs\}@cs.bu.edu
}

\maketitle

\begin{abstract}
We consider computations of a Turing machine under
noise that causes  consecutive violations
of the machine's transition function.
Given a constant upper bound $\beta$ on the size of bursts of faults,
we construct a Turing machine $M(\beta)$ subject to faults
that can simulate any fault-free
machine under the condition that bursts are not closer to each other than
$V$ for an appropriate $V = O(\beta^2)$. 
\end{abstract}

\section{Introduction}

\subsection{The  problem}
Little is known about the behavior and the power of Turing
machines when their operation is subjected to noise
that can change arbitrarily the state and the content of
the cell where the head is positioned.
The main open question, under every noise model,
is whether a machine subject to it can
perform arbitrary computations reliably.

Here, we construct a Turing machine that---with a
slowdown by a multiplicative constant---can simulate any other Turing
machine even if the simulator is subjected to constant size bursts
of faults separated by a certain minimum number
of steps from each other.

The problem of constructing fault-proof machines
from components that can fail was first considered
by von Neumann in \cite{vonNeumann},
who addressed the problem in the Boolean circuits model.
New advances along this path were made in~\cite{NC,DSpielman}.
The question has been considered in uniform models
of computation as well.
A simple rule for two-dimensional
cellular automata that keeps one bit forever even though
each cell can fail with some small probability
was given in~\cite{Toom}.
A 3-dimensional reliably computing cellular automaton using Toom's rule
was constructed in~\cite{PGReif}.
Alas, all simple one-dimensional cellular automata
appear to be ``ergodic'' (forgetting everything about their initial
configuration in time independent of the size).
The first, complex, nonergodic cellular automaton
was constructed in~\cite{PG:1986}, and improved upon
in~\cite{GacsSorg97}.
It supports a hierarchical organization,
based on an idea given in~\cite{Kurd}.
Cells are organized in
units that perform fault-tolerant simulation of another
automaton (of the same kind).
The latter simulates even more reliably a third
automaton of a similar kind, and so on.

The question of reliable computation with
Turing machines (where arbitrarily large
bursts may occur with correspondingly small probability) is
raised in~\cite{GacsSorg97}.
As in the case of one-dimensional cellular automata,
no simple solution to this problem appears to exist.
The present paper's machine is intended as a building block
towards the eventual (hierarchical) construction of such a
machine.
This follows the paradigm of the proof in~\cite{PG:1986}, where
each member of the hierarchy of simulations
is a similar building block, coping with distant bursts.
To the best of our knowledge, this is the first
construction of a sequential machine, reliable in a similar sense.

The title of~\cite{asarin} suggests some connection, but that
paper's interest is completely different: it examines the expressional ability,
in terms of the arithmetical hierarchy, of Turing
machines whose storage tapes are exposed
to stochastic noise that tends to zero.

\subsection{Simulating cellular automata}

It is natural to try to derive fault-tolerant Turing
machines from the existing results on fault-tolerant cellular automata.
Cannot one simply define a Turing machine that simulates a fault-tolerant
cellular automaton?
In some sense, the answer is yes.
Suppose that we know in advance the memory requirement
$m = S(x)$ of a computation on a
fault-tolerant cellular automaton $M$, on input $x$.
The we can define a special Turing machine $T(m)$, working on a \df{circular tape} of
size $m$, with the head moving always in the right direction (in other words,
the ``oblivious'' property is hardwired), where each pass of $T$ over the tape
simulates one step of $M$.
This machine will clearly have the same fault-tolerance
properties that $M$ has.

The circular Turing machine has a
strong fault-tolerant behavior (with a sophisticated transition rule, coming from the
cellular automaton it simulates).
Our efforts on fault-tolerant Turing machines
can be seen as just aiming to
remove the limitation of circular tape
(input size-dependent hardware).

In view of the above, it would be sufficient to
define fault-tolerant sweeping
behavior on a regular tape
(once the head can change direction, the sweeping
movement can be disturbed by faults):
the rest can be done by simulating a cellular automaton.
We were, however, not able to do this without
recreating the hierarchical
constructions used in cellular
automata---with all the necessary changes for
Turing machines.

\subsection{Turing machines}

Our contribution uses one of the standard definitions of a Turing
machine, with the exception of no halting state.

\begin{definition}
    A Turing machine $M$ is defined by a tuple
        \begin{align*}
             (\Gamma,\Sigma,\delta,q_{\start},F).
        \end{align*}
    Here, $\Gamma$ is a finite set of \df{states}, $\Sigma$ is a finite alphabet used
    in cells of the tape, and
        \begin{align*}
             \delta: \Sigma\times \Gamma\to \Sigma\times\Gamma\times\{-1,0,+1\}
        \end{align*}
    is a transition function.
    The tape alphabet $\Sigma$ contains at least the distinguished
    symbols $\blank,0,1$ where $\blank$ is called the \df{blank symbol}.
    The distinguished state $q_{\start}$ is called the \df{starting state}.
    The set $F$ of \df{final states} has the property that whenever $M$ enters a
    state in $F$, it can only continue from there to another state in $F$, without
    changing the tape.

    A \df{configuration} is a tuple
        \begin{align*}
             (q,\h,x),
        \end{align*}
    where $q\in\Gamma$, $\h\in\bbZ$ and $x\in\Sigma^{\bbZ}$.
    Here, $x[p]$ is the content of the tape cell at position $p$.
    The tape is blank at all but finitely many positions.
    The work of the machine can be described as a sequence of configurations
    $C_{0},C_{1},C_{2},\dots$, where $C_{t}$ is the configuration at time $t$.
    If $C=(q,\h,x)$ is a configuration then we will write
        \begin{align*}
             C.\state=q,\quad C.\pos=\h, \quad C.\tape=x.
        \end{align*}
    Here, $x$ is also called the \df{tape configuration}.
    \end{definition}

Though the tape alphabet may contain
non-binary symbols, we will restrict input and output to binary.

    \begin{definition}
For an arbitrary binary string $x$, let
        \begin{align}\label{eq:M()}
            M(x,t)
        \end{align}
    denote the configuration at time $t$, when started from a binary input string $x$
    written on the tape starting from position 0, with head position 0 and the
    starting state.
    Thus, the symbol at tape position $p$ at time $t$ can be written
        \begin{align*}
             M(x,t).\tape[p].
        \end{align*}
    The transition function $\delta$ tells us how to compute the next
    configuration from the present one.
    When the machine is in a state $q$, at tape position $\h$, and
    observes tape cell with content $a$, then denoting
         \begin{align*}
           (a',q',j)=\delta(a,q),
         \end{align*}
    it will change the state to $q'$, change the
    tape cell content to $a'$ and move to tape position to $\h+j$.
    For $q\in F$ we have $a'=a$, $q'\in F$.
\end{definition}

We say that a \df{fault} occurs at time $t$ if the output $(a',q',j)$ of the
transition function at this time is replaced with some other value
(which is then used to compute the next configuration).
For the sake of a clean definition of simulations, we will be more formal in
defining fault-free histories.

\begin{definition}[Trajectory]\label{def:configs,histories}
    Let
        \begin{align*}
             \Configs_{M}
        \end{align*}
    denote the set of all possible configurations of a Turing machine $M$.
    Consider a sequence $\eta=(\eta(0),\eta(1),\dots)$ of configurations of
    $M=(\Gamma,\Sigma,\delta)$ with $\eta(t)=(q(t),\h(t),x(t))$.
    This sequence will be called a \df{history} of $M$
    if  the following conditions hold:
    \begin{bullets}
        \item $q(0)=q_{\start}$.
        \item $x(t+1)[n]=x(t)[n]$ for all $n\ne \h(t)$.
        \item $\h(t+1)-\h(t)\in\{-1,0,1\}$.
    \end{bullets}
    Let
        \begin{align*}
            \Histories_{M}
        \end{align*}
    denote the set of all possible histories of $M$.
    A history $\eta$ with $\eta(t)=(q(t),\h(t),x(t))$ of $M$ is called a
    \df{trajectory} of $M$  if for all $t$ we have
        \begin{equation}\label{eq:traj}
             (x(t+1)[\h(t)], q(t+1), h(t+1)-h(t)) = \delta(x(t),q(t)).
        \end{equation}
    We say that a history has a \df{fault} at time $t$ if~\eqref{eq:traj} is violated
    at time $t$.
(Thus, if a history has any one fault, it is not a trajectory.)
A \df{burst of faults} of a history is a sequence of times containing some faults.
\end{definition}

With the earlier notation~\eqref{eq:M()}, if $x\in\Sigma^{*}$ is a string of nonblank tape
symbols, then the history defined by
\begin{align*}
 \eta(t)=M(x,t)
\end{align*}
for all $t$ is a trajectory in which $\eta(0)$ is a starting tape configuration
obtained by surrounding $x$ with blanks.

\subsection{Codes}

To define simulation of a noise-free machine $M_{2}$ by a noisy machine $M_{1}$,
we need to specify the correspondence between
configurations of these two machines.
After a burst, the state of the machine---as well as the state of cells
where the head was during the burst, could have been changed in an arbitrary way.
To proceed with the simulation, the simulating machine must
recover the information lost.
Redundant storage will help.
In Section~\ref{sec:simulation}, we will specify how one step of
$M_{2}$ is simulated by a bounded number of steps of $M_{1}$.

We formalize redundant storage with the help of codes.

\begin{definition}[Code]\label{def:codes}
    Let $\Sigma_{1},\Sigma_{2}$ be two finite alphabets.
    A \df{block code} is given by a positive integer $Q$, an \df{encoding
      function}  $\varphi_{*} :\Sigma_{2}\to\Sigma_{1}^{Q}$ and a \df{decoding
      function} $\varphi^{*}:\Sigma_{1}^{Q}\to\Sigma_{2}$
 with the property $\varphi^{*}(\varphi_{*}(x))=x$.
 \end{definition}

 \begin{definition}[Standard pairing]
For every alphabet $\Sigma$ that we will consider, we assume that there is a
standard ordering of its elements: $\Sigma=\{s_{1},\dots,s_{n}\}$.
This gives rise to a code
 \begin{align*}
(\gamma_{*},\gamma^{*}),
 \end{align*}
where $\gamma_{*}(s_{i})$ is the base 2 notation of the number $i$, padded
from the front to length $\cei{\log n}$.
For example, if $\Sigma=\{s_{1},s_{2},s_{3}\}$
then the codewords are $01,10,11$.

For a (possibly empty) binary string $x=x(1)\dotsm x(n)$ let us introduce the map
 \[
   x^{o} = x(1)x(1)x(2)x(2)\dotsm x(n)x(n).
 \]
If $s$ is a symbol in some alphabet $\Sigma$ then by $\ang{s}$ we will understand
$(\gamma_{*}(s))^{o}$, and call it the \df{standard prefix-free code} of $s$.
Similarly,
 \begin{align*}
        \ang{s,t} &=\Paren{(\gamma_{*}(s))^{o}\gamma_{*}(t)}^{o},
\\ \ang{s,t,u} &= \ang{s,\ang{t,u}},
 \end{align*}
and so on.
\end{definition}

We have $|x^{o}|=2|x|+2$.
There are shorter codes with the same prefix-free property,but minimizing the
code length is not our concern here.

 \begin{definition}[Error-correcting code]\label{def:ercode}
A block code is $(\beta,t)$-\df{burst-error-correcting}, if for all
$x\in\Sigma_{2}$, $y\in\Sigma_{1}^{Q}$ we have $\varphi^{*}(y)=x$
whenever $y$ differs from $\varphi_{*}(x)$ in at most $t$
intervals of size $\le\beta$.
 \end{definition}

 \begin{example}[Tripling]\label{xmp:tripling}
   Suppose that $Q\ge 3\beta$ is divisible by 3,
$\Sigma_{2}=\Sigma_{1}^{Q/3}$, $\varphi_{*}(x)=xxx$.
Let $\varphi^{*}(y)$ be obtained as follows.
If $y=y(1)\dots y(Q)$, then $x=\varphi^{*}(y)$ is defined as follows:
$x(i)=\maj(y(i),y(i+Q/3),y+2Q/3)$.
For all $\beta\le Q/3$, this is a 
$(\beta,1)$-burst-error-correcting code.

If we repeat 5 times instead of 3, we get a $(\beta,2)$-burst-error-correcting
code (there are also much more efficient such codes than just repetition).
 \end{example}

We will also need a more general majority function later on:

    \begin{definition}\label{def:majority}
        Let $x=(x_{1},\dots,x_{n})$ be a sequence of symbols
        from a finite alphabet $\Sigma=\{a_{1},a_{2},\dots,a_{m}\}$.
        For each $j=1,2,\dots,m$, let  $k_{j}$ be the number of
        occurrences of $a_{j}$ in $x$, $k_{1} + k_{2} + \dots + k_{m} = n$.
        Then,
            \begin{align*}
                 \maj(x_{1},x_{2},\dots,x_{n}) = a_{k},
            \end{align*}
        where  $k = \arg\max_{j}{k_{j}}$.
    \end{definition}

    In the Section~\ref{sec:majority}, we will show how to compute
    a majority of values in a sequence of cells using only one
    pass over the sequence.

 \begin{definition}[Configuration code]
      A \df{configuration code}
is a pair of functions
    \begin{align*}
        \varphi_{*} :\Sigma_{2}^{\bbZ}\to\Sigma_{1}^{\bbZ},
        \quad
        \varphi^{*}:\Sigma_{1}^{\bbZ}\to\Sigma_{2}^{\bbZ}
    \end{align*}
    that encodes infinite strings of $\Sigma_{2}$ into infinite strings of $\Sigma_{1}$.
    Each block code $(\varphi_{*},\varphi^{*})$
    gives rise to a natural tape configuration code which we will also denote by
    $(\varphi_{*},\varphi^{*})$.
    If $\xi=\dotsm \xi(-1)\xi(0)\xi(1)\dotsm$ is an infinite string of letters of $\Sigma_{2}$
    then $\varphi_{*}(x)$ is the string
    \begin{align*}
        \dotsm\varphi_{*}(\xi(-1))\varphi_{*}(\xi(0))\varphi_{*}(\xi(1))\dotsm,
    \end{align*}
    while for decoding an infinite configuration $\xi'$ we subdivide it first into
    blocks of size $Q$ (starting with $\xi'(0)\dotsm\xi'(Q-1)$), decode each block
    separately, and concatenate the results.
\end{definition}

\subsection{The result}

We will define our result in terms of universal Turing machines, operating on
binary strings as inputs and outputs.

\begin{definition}[Computation result]
 Assume that a Turing machine $M$ starting on binary $x$,
 at some time $t$
 arrives at the first time at some final state.
 Then we look at the longest (possibly empty)
 binary string to be found starting at position
 0 on the tape, and call it the \df{computation result} $M(x)$.
\end{definition}

\begin{definition}[Universal Turing machine]\label{def:univ-TM}
We say that Turing machine $U$ is \df{universal} among Turing machines with
binary inputs and outputs, if for every Turing machine $M$,
for all binary strings $x$, there is a binary string $p_{M}$ such that $M$ reaches a
final state on input $x$ if and only if $U$ reaches a final state on input $\ang{p_{M}}x$,
further in this case we have $U(\ang{p_{M}}x)=M(x)$.
A universal Turing machine will be called \df{flexible} if whenever
$U(\ang{p,q}x)$ halts, also $U(\ang{p}q)$ halts, and
 $U(\ang{p,q}x)=U(U(\ang{p}q)x)$.
In other words  if a program has the form $\ang{p,q}$ then $U$ first applies
as a preprocessing step the program $p$ to $q$, and then it starts work on the
result attached in front of $x$.

It is well-known that there are flexible universal Turing machines.
Let us fix one and call it $U$.

Consider an arbitrary Turing machine $M$ with state set $\Gamma$, alphabet
$\Sigma$, and transition function $\delta$.
A binary string $p$ will be called a \df{transition program} of $M$ if
whenever $\delta(a,q)=(a',q',j)$ we have
 \begin{align*}
 U(\ang{p}\ang{a,q})=\ang{a',q',j}.
 \end{align*}
We will also require that the computation induced by the program makes
$O(|p|+|a|+|q|)$ left-right turns, over a length tape $O(|p|+|a|+|q|)$.
\end{definition}

The transition program just provides a way to compute
the (local) transition function of $M$ by the universal machine,
it does not organize the rest of the simulation.

\begin{remark}
 In the construction provided by the textbooks,
 the program is generally a string
encoding a table for the transition
function $\delta$ of the simulated machine
$M$.
Other types of program are imaginable:
some simple transition functions can
have much simpler programs.
However, our fixed machine is good enough.
Let the fixed program $r$ be such that
$U(\ang{r}\ang{x,y})=\ang{x}\ang{y}$.
If some machine $U'$ simulates $M$ via a
very simple program $q$, then
$U$ will simulate $M$ via program $\ang{r,\ang{p_{U'},q}}$:
 \begin{align*}
     U(\ang{r,\ang{p_{U'},q}}x) =
     U(\ang{p_{U'}}\ang{q}x)=
     U'(\ang{q}x) = M(x).
 \end{align*}
\end{remark}

For simplicity,
we will consider only computations whose result
is a single symbol, at tape position 0:
\begin{align*}
 M(x,t).\tape[0]
\end{align*}
at any time $t$ in which $M(x,t).\state$ is a final state.
This frees us of the problem of having to
decode before announcing the final
result of the fault-tolerant computation.
We will prove:

\begin{theorem}[Main]\label{thm:main}
For a given Turing machine $M_{2}$ with transition program $p_{2}$,
and positive integer $\beta$, following items can be constructed:
\begin{bullets}
    \item Integers $Q$ depending linearly on $\beta$ and $p_{2}$,
        logarithmically on $|\Sigma_{2}|$, $|\Gamma_{2}|$, further
        $V$ depending quadratically on $Q$;
    \item A block code $(\varphi_{*},\varphi^{*})$ of blocksize $Q$;
    \item A machine $M_{1}$
        whose number of states and alphabet
        size depend polynomially on $Q$,
        with some function $f$ defined on its alphabet;
\end{bullets}
such that the following holds.

Suppose that on input $x$, the fault-free machine $M_{2}$
enters a final state at time $T$.
Assume that $\eta$ is a history of machine $M_{1}$ on starting
configuration $\varphi_{*}(x)$
such that bursts of faults have size at most
$\beta$, and are separated
by at most $V$ steps from each other.
Let $t$ be any time $\ge V T$ such that no fault occurred in
the last $V$ steps before and including $t$, then
\begin{align}\label{eq:main-thm}
 f(\eta(t).\tape[0])= M_{2}(x,T).\tape[0].
\end{align}
\end{theorem}

Section~\ref{sec:structure} specifies the layout of
the tape and the structure of
the states, and introduces the notion of rules.
The parts of the transition function of $M_{1}$
dealing with redundant simulation are
defined in Section~\ref{sec:simulation}.
Section~\ref{sec:faults} introduces the parts
allowing to restore the structure of the simulation
after locally garbled by faults.
The main theorem is proved in Section~\ref{sec:proof}.

\section{Program Structure}\label{sec:structure}

\subsection{Fields}

Each state of the simulating machine $M_{1}$ will
be a tuple $q=(q_{1},q_{2},\dots,q_{k})$,
where the individual elements of the tuple will be called \df{fields}, and will
have symbolic names.
For example, we will have fields $\Info$ and $\Drift$,
and may write $q_{1}$ as $q.\Info$ or just $\Info$, $q_{2}$ as $q.\Drift$ or $\Drift$, and so on.

We will call the current direction
of the simulated machine $M_{2}$ the \df{drift} ($-1$ for left, 0 for none, and $+1$ for right).

A properly formatted configuration of $M_{1}$ splits the tape into blocks of $Q$
consecutive cells called \df{colonies}.
One colony of the tape of the simulating
machine represents one cell of the simulated machine.
The colony that
corresponds to the active cell of the simulated machine (that is the cell that the
simulated machine is scanning) is called the \df{base colony}
(later we will give a precise definition of this notion based on the actual history of
the work of $M_{1}$).
Once the drift is known, the union of the base colony with the neighbor colony in
the direction of the drift is called the \df{extended base colony}
(more precisely, see Definition~\ref{def:healthy}).

The head will make some global sweeping movements over the extended base
colony.
We will use term \df{sweep direction}.
But even while the sweep direction does not change,
the head will make frequent short switchbacks (zigzags).

The states of machine $M_{1}$ will have a field called \df{mode}.
The normal mode corresponds to the states
where $M_{1}$ is performing the simulation of $M_{2}$.
The recovery mode tries to correct some perceived fault.

Similarly, each cell of the tape of $M_{1}$ consists of several fields.
Some of these have names identical to fields of the state.
In describing the transition rule of $M_{1}$ we will write, for example,
$q.\Info$ simply as $\Info$, and for the corresponding field $a.\Info$ of the
observed cell symbol $a$ we will write $\cInfo$.
The array of values of the same field of the cells will be called a \df{track}.
Thus, we will talk about the $\cHold$ track of the tape, corresponding to the
$\cHold$ field of cells.

\subsection{Rules}

Instead of writing a single huge transition table, we present the transition
function as a set of \df{rules}.
Each rule consists of some conditional statements, similar to the ones seen in
an ordinary program:
 ``{\bf if }{\em condition} {\bf then}$\dots$'', where the condition
is testing values of some fields of the state and the observed cell.
Even though rules are written like procedures of a program,
they describe a single transition.
When several consecutive statements are given, then they
(almost always) change different fields of the state or
cell symbol, so they can be executed simultaneously.
Otherwise and in general, even if a field is updated in
some previous statement, in all following statements that use
this field, its old value is considered.

Rules can call other rules, but these calls will never form a cycle.
We will also use some conventions introduced by the C language:
namely,
$x\gets x+1$ and $x\gets x-1$ are abbreviated to $\increment{x}$ and
$\decrement{x}$ respectively.

Rules can also have parameters, like $\ruSwing(a,b,u,v)$ (see below).
Since each rule is called only a constant number of times in the whole program,
the parametrized rule can be simply seen as a shorthand.

\subsection{List of fields}
The basic fields of the state and of cells are listed below, with some hints of
their function (this does not replace our later definition of the transition function).
We will not repeat every time that each field of a cell has also a possible value
$\emptyset$ corresponding to the case when the state is blank.

\emph{We recommend to skip this list at first reading, and to return to it just for reference.}

\begin{enumerate}[1.]

\item $\Addr$ ranges from $-Q$ to $2Q-1$.
The values $-Q$ to $-1$ are taken during the left drift, while the values $Q$ to
$2Q-1$ during a right drift.
During the first sweep of the work period, $\Addr$ is reduced modulo $Q$.

\item $\Dir$ stores the direction of the previous step.

\item $\Drift$ stores the direction of the simulated machine $M_{2}$.
It may have values $\emptyset,-1,0,1$.
The value $\emptyset$ corresponds to the case when the new $\Drift$ is still not
computed,
and will also be the default value (for example in empty cells).
The $\cDrift$ field of the cells of the extended colony
correspond to $\Drift$ in the state.

\item $\Sweep$ numbers the sweeps through the colony.
The first sweep of a work period has number 1 and is to the right,
and this way each right sweep is odd, each
left sweep is even.
Thus the sweep direction of the head is completely determined by
the parity of $\Sweep$, unless the head is at a ``turning'' point.
At turning points, $\Sweep$ is incremented.

Field $\cSweep$ holds the number of the most recent sweep.
The simulation consists of a \df{computation phase} and a \df{transferring phase},
each corresponding to a certain
interval of sweep values to be specified below (these intervals depend
somewhat on the value $\Drift$).

\item The triple of fields $(\cAddr, \cSweep,\cDrift)$ will
determine the role played by a cell in the colony work period: for
notational convenience, we introduce the names
 \begin{align}\label{eq:Core}
   \Core = (\Addr,\Sweep,\Drift),
\quad  \cCore = (\cAddr,\cSweep,\cDrift).
 \end{align}

\item
 The $\Info$ and  $\State$ tracks represent the tape symbols and the state
of the simulated machine $M_{2}$.

\item\label{i:fields.zigging}
 The sweep-through is interrupted by switchbacks called \df{zigging},
described by a rule $\Zigzag(d)$.
Here $d$ is the direction of sweep.
The process also depends on a fixed parameter
 \begin{align}\label{eq:Z}
  \Z = 22 \beta, 
 \end{align}
and is controlled by the fields $\ZigDepth$ and $\ZigDir$ of the state.

From now on, we will assume that $Q$ is a multiple of $\Z-4\beta$:
 \begin{align}\label{eq:zig|Q}
   (\Z-4\beta)|Q.
 \end{align}
Every $\Z-4\beta$ forward steps are accompanied by $\Z$ steps
backward and forward, for a total of
 \begin{align*}
 3\Z-4\beta=3(\Z-4\beta)+8\beta<4(\Z-4\beta)
 \end{align*}
steps.

\begin{figure}[!ht]
\centering
\includegraphics[height=1.8in]{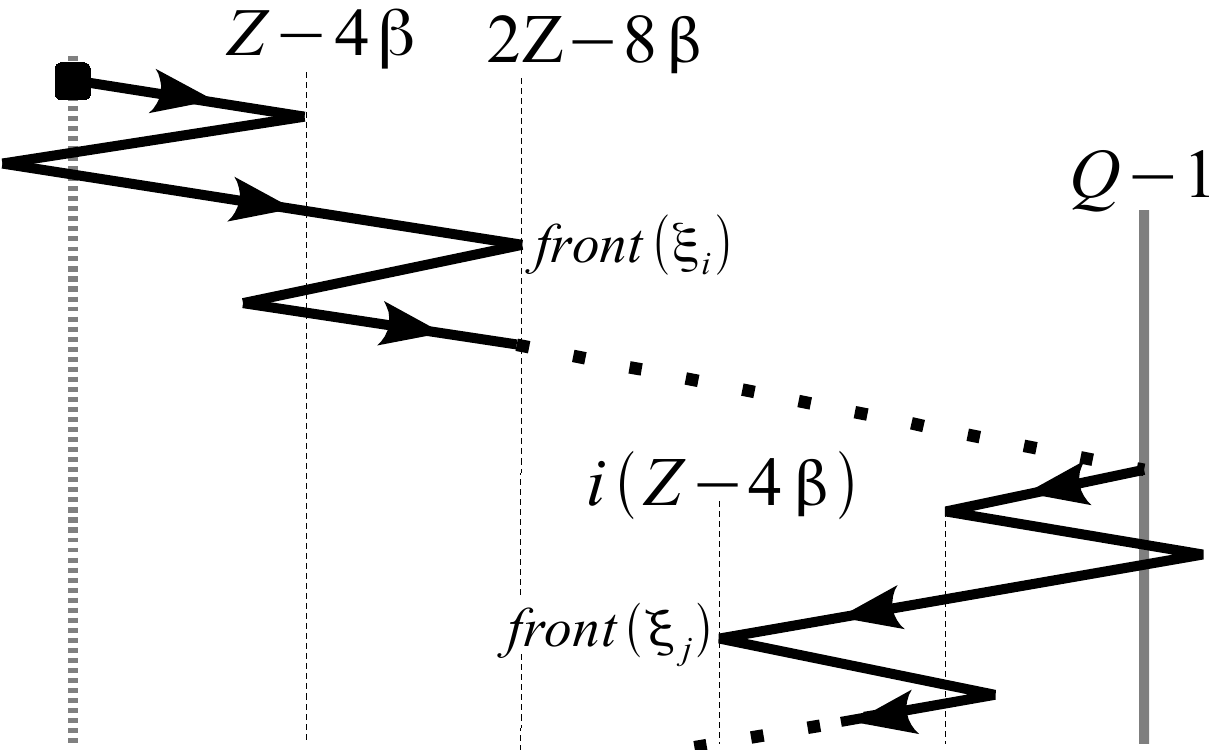}
\caption{A sweep of a base colony with zigging.
\label{fig:zigging}}
\end{figure}

\item The $\Mode$ field takes values in $\{\Normal,\Recovering\}$.
In the absence of faults, the state would never leave the normal mode.
On noticing any disorder, the state will switch to recovery mode, with the goal
of eventually returning to normal mode.

The fields used in recovery mode are all collected as subfields of the field
$\Rec$ of the state, and the field $\cRec$ of the cell state.
They will be introduced in the definition of the recovery rule.

In particular, when the field
 \begin{align}\label{eq:cRecCore}
  \cRec.\Core
 \end{align}
is not $0$, we will call the cell \df{marked}.

\item Even though we store information with redundancy, faults can disturb the
  coding and decoding operations and the simulating computation itself.
Therefore these procedures will be repeated several times, and their results,
serving as candidates of the final values of $(\Info,\Drift,\State)$, will be stored in
the $\cHold$ track.
The different candidates will be stored in the different parts of the $\cHold$
field, which is actually a small array
$\cHold[1]$, $\cHold[2]$, $\cHold[3]$.
The subfield $\cHold[i].\Info$ holds the value of the $i$-th candidate for the
new $\Info$ value of the current cell.

\end{enumerate}

Machine $M_{1}$ has no final states: even when the simulated computation ends
the simulation continues to defend the end result from faults.

\section{The Simulation}\label{sec:simulation}

One computational step of the machine $M_{2}$ is simulated by many steps of
$M_{1}$ that make one unit called the \df{work period}.

\subsection{Coding}\label{subsec:coding}

We will frequently make use of the parameter
 \begin{align}\label{eq:Expansion}
      \E = 30 \Z.
 \end{align}
For simplicity, let us assume that  the set of states $\Gamma_{2}$,
and the alphabet $\Sigma_{2}$ are subsets of the set of  binary strings
$\{0,1\}^{\ell}$ for some $\ell<Q$ (we can always ignore some states or tape
symbols, if we want).
We will then use the same code $(\upsilon_{*},\upsilon^{*})$
for both the states of machine $M_{2}$ and its alphabet.
Let $(\upsilon_{*}, \upsilon^{*})$ be a 
$(\beta,2)$-burst-error-correcting code
\begin{align*}
 \upsilon_{*}: \{0,1\}^{\ell}\to\{0,1\}^{Q-2.2\E}.
 \end{align*}
(The length of the code is not $Q$, only $Q-2.2\E$, since we will leave place the
codeword at a distance $1.1\E$ from both colony ends.)
We could use, for example, the tripling code of Example~\ref{xmp:tripling}.
Other codes are also appropriate, but we require that they have some fixed,
constant programs $p_{\encode}$, $p_{\decode}$
on the universal machine $U$, in the following sense:
 \begin{align*}
   \upsilon_{*}(x)=U(\ang{p_{\encode}}x),\quad
   \upsilon^{*}(y)=U(\ang{p_{\decode}}y).
 \end{align*}
Also, these programs must work in quadratic time and linear space on a one-tape
Turing machine (as tripling certainly does).

Let $a$ be the tape configuration of $M_{2}$ at time 0, and $s$ the starting
state of $M_{2}$.
The initial tape configuration $a'=\varphi_{*}(a)$ of $M_{1}$
is defined as follows:
\begin{align}\label{eq:init}
     a'[h \cdot Q+1.1\E,\dots,(h+1)Q-1.1\E-1].\Info   & = \upsilon_{*}(a[h]),
\\   a'[1.1\E,\dots, Q-1.1\E-1].\State  & = \upsilon_{*}(s).
\end{align}

In cells of the base colony and its left neighbor  colony,
the $\cSweep$ and $\cDrift$ fields are set
to $\Last(+1)$ and $1$ respectively, where $\Last(+1)$
denotes the last sweep of a working period for
the positive drift (and is defined below in~\eqref{eq:Last}).
In the right neighbor colony, these values are $\Last(-1)$ and $-1$
respectively.
In all other cells, these values are empty.

The head is initially located at the first cell of the base colony.
We assume that the $\Addr$
fields of each colony are filled properly, that is
\begin{align*}
   a'[i].\Addr = i \bmod Q.
\end{align*}
The $\cHold$ values are empty in each cell.

Machine $M_{1}$ starts in normal mode,
$\Drift=1$, $\Sweep=1$.
All other fields have also their initial (or empty)
values (see Figure~\ref{fig:coding}).

\begin{figure}[!ht]
\centering
\includegraphics[width=6in]{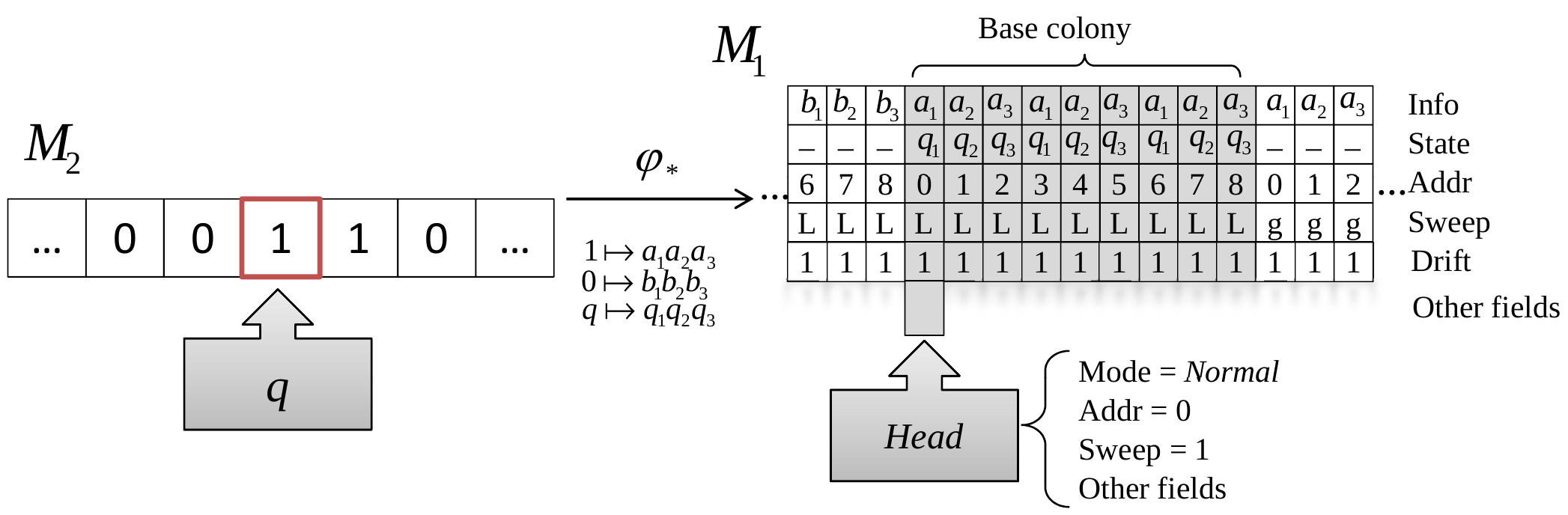}
\caption{The initial configuration of
machine $M_{2}$ is encoded into the initial configuration
of $M_{1}$, where $L=\Last(1)$, $g=\Last(-1)$. \label{fig:coding}
}
\end{figure}

The corresponding block decoding function
$\varphi^{*}$ is obtained applying the decoding
function $\upsilon^{*}$ to just the $\cInfo$ track of $M_{1}$ (actually to just
the part between addresses $1.1\E$ and $Q-1.1\E$ of each colony) to obtain the
tape of the simulated machine
$M_{2}$, and to just the  $\cState$ track of the base colony to obtain its state.

This definition of decoding will be refined for configurations different from
the initial one, since the
location of the base colony must also be found using decoding.

\subsection{Sweep counter and direction}

The global sweeping movement of the head will be controlled by the parametrized rule
\begin{align*}
\ruSwing(a,b,u,v).
\end{align*}
This rule makes the head swing between two extreme points $a,b$,
while the counter $\Sweep$ increases from value $u$ to value $v$.
The $\Sweep$ value is incremented at the ``turns'' $a,b$ (and is
also recorded on the track $\cSweep$).

The sweep direction $\delta$ of the simulating head is derived from
$\Sweep$, $\Addr$ and the current value $\Dir$ in the following way.
On arrival of the head to an endpoint (that is
when $\Dir \neq 0$ and $\Addr\in\{a,b\}$), the values
$\Sweep$ and $\cSweep$ are incremented and $\Dir$ is set to 0.
In all other cases, the sweep direction is determined
by the formula
 \begin{align}\label{eq:sweep-dir}
   \dir(s)=(-1)^{s + 1}.
 \end{align}
Let
\begin{equation}\label{eq:Dir}
\delta =
  \begin{cases}
     0 & \text{ if $\Addr\in\{a,b\}$ and $\Dir\neq 0$},  \\
       \dir(\Sweep) &\text{otherwise}.
     \end{cases}
\end{equation}
As an example of rules, we present the zigging rule in
Rule~\ref{alg:Zigzag}, which itself uses the rule $\Move(d)$.
At each non-zigging step, $\Addr \gets \Addr+\delta$.

\begin{algorule}[!ht]\caption{$\Move(d)$}\label{alg:Move}

  $\Dir\gets d$\Comment*[f]{$d\in\{-1,0,1\}$.}\;

  \lIf{$\Mode=\Normal$}{$\Addr\gets\cAddr\gets\Addr+d$}\;
  \lElse{$\Rec.\Addr\gets\Rec.\Addr+d$}\;
  Move in direction $d$.
\end{algorule}

\begin{algorule}[!ht]\caption{$\Zigzag(d)$}\label{alg:Zigzag}
  \Comment{$d\in\{-1,1\}$ is the direction of progress.}
  \If{$\ZigDir=-1$ \algAnd (($\ZigDepth=0$ \algAnd $(\Z-4\beta)|\Addr$)
    \\ \algOr $0<\ZigDepth<\Z$)}{
    $\increment{\ZigDepth}$\;
    \lIf{$\ZigDepth = \Z-1$ 
    }{$\ZigDir=1$}\;
    $\Move(-d)$
  }
  \ElseIf{$\ZigDir=1$ \algOr ($\ZigDepth=0$ \algAnd $(\Z-4\beta)\not|\Addr$)}{
    \lIf{$\ZigDepth>0$}{$\decrement{\ZigDepth}$}
    \lElse{$\ZigDir\gets -1$}\;
    $\Move(d)$
  }
\end{algorule}

\subsection{The computation phase}\label{subsec:computation-phase}

The aim of this phase is to obtain new values for $\cState$, $\cDrift$ and
$\cInfo$.
It essentially repeats three times the following \df{stages}:
\df{decoding}, \df{applying the transition}, \df{encoding}.
In more detail:
\begin{enumerate}

      \item For every $j=1,2,3$ do
       \begin{enumerate}
          \item Calling by $g$ the  string found on the $\cState$ track of
            the base colony between addresses $1.1\E$ and $Q-1.1\E$, decode it
            into string $\hat g=\upsilon^{*}(g)$
            (this should be the current state of the simulated machine $M_{2}$), and
            store it on some auxiliary track in the base colony.
            Do this by simulating the universal machine on some work track, with
            the program $p_{\decode}$: $\hat g = U(\ang{p_{\decode}}g)$.

            Proceed similarly with the string $a$ found on the $\cInfo$
            track of the base colony, to get $\hat a = \upsilon^{*}(a)$
            (this should be the observed tape symbol of the simulated machine $M_{2}$).

          \item Compute the value
  \begin{align*}
    (a',g',d)=\delta_{2}(\hat a, \hat g)
  \end{align*}
            similarly, simulating the universal machine $U$ with program $p_{2}$.
            The string $p_{2}$ (as well as the constant-size programs
            $p_{\decode},p_{\decode}$) is ``hardwired'' into the
            transition function of $M_{1}$, that is the program we are writing.
            More precisely, before performing the computation of
            $U(\ang{p_{2}}\ang{\hat a,\hat q})$ of Definition~\ref{def:univ-TM},
            machine $M_{1}$ writes the program $p_{2}$ onto some work track:
            for this, the string $p_{2}$ will be a literal part of the program of $M_{1}$.

            \item\label{i:comp.write}
              Write the encoded new state $\upsilon_{*}(g')$ onto the
              $\cHold[j].\State$ track of the base colony between positions
              $1.1\E$ and $Q-1.1\E$.

              Similarly, write the encoded new observed cell content $\upsilon_{*}(a')$ onto the
              $\cHold[j].\Info$ track of the base colony.
              Write also the first symbol of $a'$ into position 0 of the same track
              (just because the Main
              Theorem~\ref{thm:main} expects the result of the whole computation
              at tape position 0, not $1.1\E$).

              Write $d$ into the $\cHold[j].\Drift$ field of \emph{each cell} of
              the base colony.
      \end{enumerate}
      \item Repeat the following twice:
          \begin{quote}
            Sweeping through the base colony,
            at each cell compute the majority of $\cHold[j].\Info$, $j=1,2,3$,
            and write into the field $\cInfo$.
            Proceed similarly, and simultaneously, with $\State$ and $\Drift$.
          \end{quote}

It can be arranged---and we assume so---that the total number of sweeps of this
phase, and thus the starting sweep number of the next phase,
 \begin{align}\label{eq:TransferStart}
\TransferStart=O(Q),
 \end{align}
depends only on $Q$.
\end{enumerate}

\begin{figure}[!ht]
\centering
\begin{tabular}{ccc}
\includegraphics[scale=0.55]{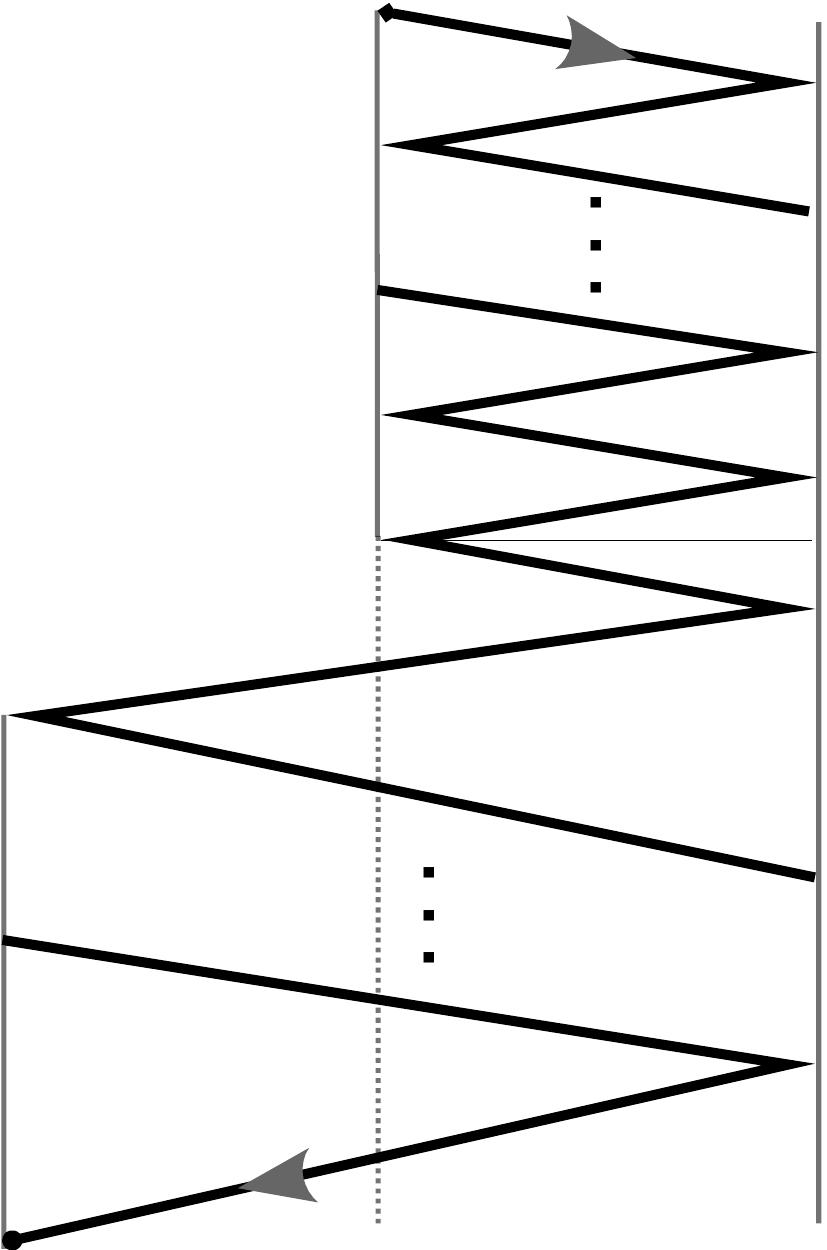} &\hspace{3em} &
\includegraphics[scale=0.55]{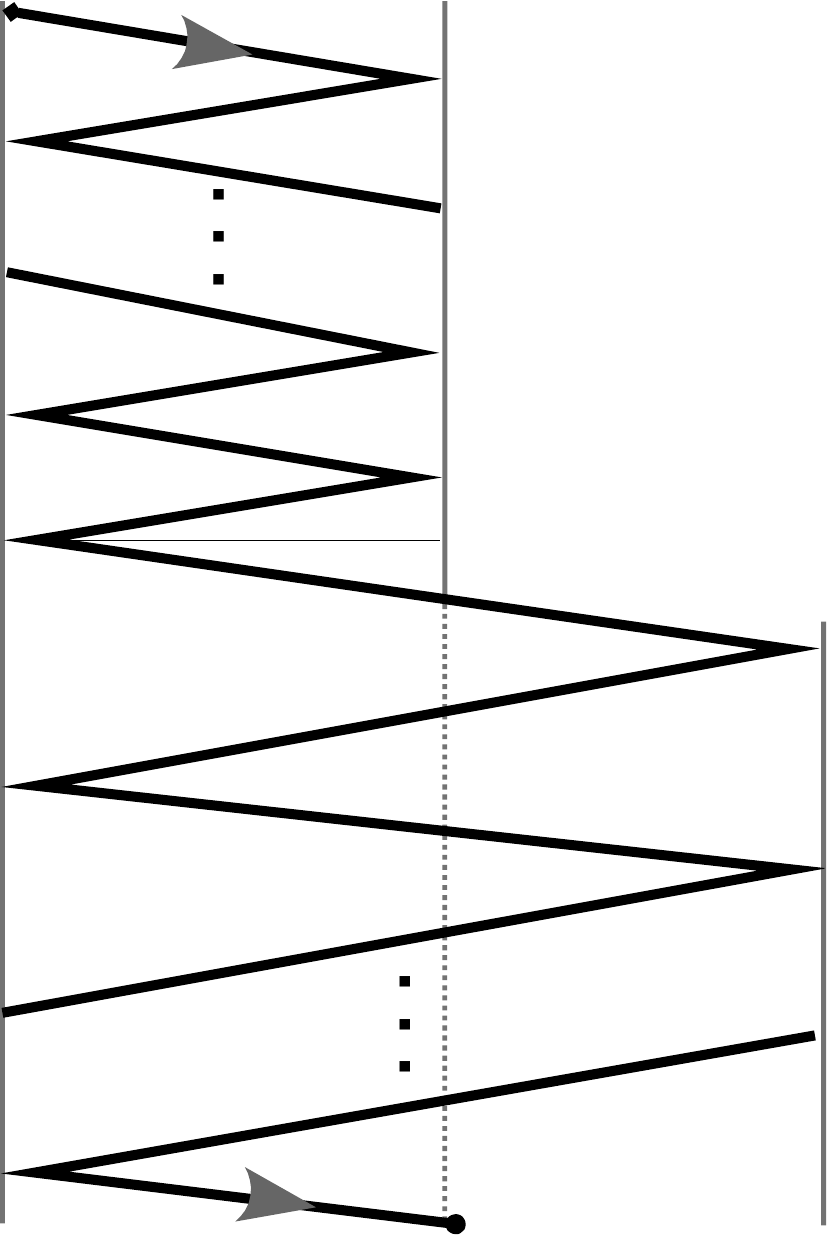} \\
(a) & & (b)
\end{tabular}
\caption{A work period of the machine $M_{1}$ (zigging and many sweeps are
not shown, for clarity).
In (a) the machine drifts left, while in (b) it drifts right.
\label{fig:sleft}}
\end{figure}

\subsection{Transfer phase}

The aim of this phase, present only if $\Drift \ne 0$,
is to transfer the new $\State$ of $M_{2}$ into the neighbor
colony in the direction of $\delta=\Drift$
(which was computed in the previous phase),
and to move there.
$\TransferSw(\delta)$ is the \df{transfer sweep}, the sweep in which we start
transferring in direction $\delta$:
\begin{equation}\label{eq:TransferSw}
                  \TransferSw(1) =\TransferStart,
\quad                 \TransferSw(-1) =\TransferStart+1.
\end{equation}
The phase consists of the following actions.
\begin{enumerate}[1.]
\item
  Spread the value $\delta$ found in the cells of the $\cDrift$ track
  (they should all be the same)
  onto the neighbor colony in direction $\delta$.

\item For $i=1,2,3$:
        \begin{quote}
          Copy the content of $\cState$ track of the base colony
            to the $\cHold[i].\State$ track of the neighbor colony.
        \end{quote}

\item Assign the field majority: $\cState\gets \maj(\cHold[1 \dots  3].\State)$
in all cells of the neighbor colony.
This part ends with a sweep value $\TransferLast$ depending only on the program
$p_{2}$.

\item If $\Drift = 1$, then move right to cell $Q$
(else stay where you are).
\end{enumerate}

The last sweep number of the work period is
\begin{align}\label{eq:Last}
  \Last(\delta) = \TransferLast + \max(0,\delta).
\end{align}

The work period in case of both non-zero drift values
is illustrated in Figure~\ref{fig:sleft}.

\subsection{Interval plurality}\label{sec:majority}

We give an algorithm that computes the \df{plurality} of some field
$\fld{c.F}$ over some
interval, that is, the value that appears the most, but at least
$1/3$ of the times.
Rule~\ref{alg:MAJ} is a version of an algorithm from \cite{misra}.
Running in a single sweep, the rule maintains a
data structure of $2$ pairs of $(v_i,c_i)$ that store
some candidate majority values and their current weight.

\begin{algorule}[htb!]\caption{$\MAJ(\fld{F},\fld{G}, n)$}\label{alg:MAJ}
\Comment{Interval ``majority''
of the field $\fld{F}$, computed and then
stored in the field $\fld{G}$ of the machine's state.
Initially $(v_i, c_i)=(\emptyset, 0)$, $i=1,2$.
}
\If{the end of the interval of length $n$ is reached}{
  $i \gets \arg \max_{j=1,2}(c_j)$\;
  $\fld{G} \gets v_i$
}
\Else{
  \If{$v_j=\fld{c.F}$ \algOr $c_{j}=0$ for some $j\in\{1,2\}$}{
    $v_j \gets \fld{c.F}$, $\increment{c_j}$
  }
  \lElse{$\decrement{c_1}$, $\decrement{c_2}$}\;
  (move right) \;
\Comment{Actually, the swing rule will move the head (with zigging).}
}
\end{algorule}

\subsection{Simulation with no faults}

The proof of Theorem~\ref{thm:main} uses a
simulation of machine $M_{2}$ by machine $M_{1}$.
Though the theorem only speaks about
the end result, for the sake of the proof
we give a formal definition of simulation.
For the moment, we concentrate on the
fault-free case.

\begin{definition}\label{def:simulation}
Let $M_{1},M_{2}$ be two machines, further let
\begin{align*}
\varphi_{*}:\Configs_{M_{2}}\to\Configs_{M_{1}}
\end{align*}
be a mapping from configurations of $M_{2}$ to those of $M_{1}$,
such that it maps starting configurations into starting configurations.
We will call such a map a \df{configuration encoding}.
Let
\begin{align*}
\Phi^{*}:\Histories_{M_{1}} \to \Histories_{M_{2}}
\end{align*}
be a mapping.
The pair $(\varphi_{*},\Phi^{*})$
of mappings is called a \df{simulation} (of $M_{2}$ by $M_{1}$)
if whenever $\xi$ is an initial
configuration of $M_{2}$ and $\eta$ is a trajectory of machine $M_{1}$
with initial configuration $\varphi_{*}(\xi)$,
the history $\Phi^{*}(\eta)$ is a trajectory of machine $M_{2}$.

We say that $M_{1}$ \df{simulates} $M_{2}$ if there is a simulation
$(\varphi_{*},\Phi^{*})$ of $M_{2}$ by $M_{1}$.
\end{definition}

We summarize the construction of the previous section
in the following statement.

\begin{lemma}\label{lem:errfreesim}
Machine $M_{1}$ simulates machine $M_{2}$.
\end{lemma}

\begin{proof}
Since there are no faults interfering
with the operation of $M_{1}$,
the history of $M_{1}$ is a trajectory,
easy to break up into \df{work periods}:
intervals in which the counter $\Sweep$ is growing.
Let $\tau_{t}$ be the end of the $t$-th work period.
(Though $\tau_{t}$ is roughly proportional
to $t$, we did not make it an exact
multiple.
Such a relation would be lost in the
faulty case, anyway.)

The code $(\varphi_{*},\varphi^{*})$ is
given in Section~\ref{subsec:coding}.
The history decoding function $\Phi^{*}$ for the noise-free case is
\begin{align*}
 \Phi^{*}(\eta)(t) = \varphi^{*}(\eta(\tau_{t})),
\end{align*}
where $\varphi^{*}$ is the tape configuration
decoding function obtained from the
block code $(\varphi_{*},\varphi^{*})$.

If $t$ reaches a final state of $M_{2}$, then starting from step $\tau_{t}$,
machine $M_{1}$ will not change the state represented on
the $\cInfo$ track anymore.
\end{proof}

We will define formally later in Definition~\ref{def:state-coord}
what it means for the state to be coordinated with the observed cell.
This is always the case in the noise-free simulation, so
let us display the ``main'' rule of machine $M_1$ in Rule~\ref{alg:main1}.
Recall the definition of \df{marked} in~\eqref{eq:cRecCore}.

\begin{algorule}[!ht] \caption{{\bf Main rule}\label{alg:main1}}
  \If{$\Mode=\Normal$}{
    \lIf{\algNot $\Coordinated$ \algOr the cell is marked for recovery}{$\Alarm$}\;
    \lElseIf{$1 \leq \Sweep < \TransferStart$}{ $\Comp$ 
    }\;
    \lElseIf{$\TransferStart \le \Sweep < \Last$}{ $\Transfer$ 
    }\;
   \lElseIf{$\Last \le \Sweep $}{$\MoveBase$ 
    }
  }
\end{algorule}
Here, rule $\MoveBase$ just moves the head to the new base in case it is
not there yet.

\section{Faults}\label{sec:faults}

A fault is a violation of the transition function.
A burst of faults can change the state to an
arbitrary one, and change an interval of cells
of size $\beta$ arbitrarily.
We will call such an interval an ``island of bad
cells'' here informally, later formally.

Faults cause two kinds of change.
One is that they change the information about
the represented machine $M_{2}$.
This problem will be corrected with the help of
redundancy (encoding of the information and
repetition of computation).
The second kind of change affects the very structure
of the simulation.
These changes will be detected and corrected
locally, by the recovery rule.

When a coordination check fails, we will switch
to recovery mode.
Recovery will start with trying to identify a
small interval containing the damage.
This is followed by restoring the
``structure'' (addresses, $\cSweep$ and $\cDrift$
values) in the interval.

\subsection{Integrity}
\label{ss:comp}

Let us specify the kind of structural integrity
we expect a configuration to have.
Informally, in the absence of faults,
``outer'' cells are those outside the base colony, and even
outside the area (to be called workspace) in which the
program extends it in the transfer phase.

\begin{definition}[Outer cells]\label{def:outer-cells}
    Recall the definition of the sweep value
    $ \Last(\delta) $ from~\eqref{eq:Last}.
    For $\delta \in \{ -1,1 \}$, if a cell is nonempty and
    has $0 \le \cAddr < Q$,
    $\cDrift = \delta$,
    $\cSweep = \Last(\delta)$
    then it will be called a \df{right outer cell} if
    $\delta = -1$.
    It is a \df{left outer cell} if $\delta = 1$.
    If it is empty then it will be considered both a left outer cell and a right
    outer cell.
\end{definition}

\begin{definition}[Healthy configuration, base colony,
    extended base colony, workspace]
\label{def:healthy}
 A configuration $\xi$ is \df{healthy} if the mode is
 normal, further the following holds.

Let $d$ denote the direction of sweep, as determined by~\eqref{eq:Dir}.
Recall that $\xi.\pos$ is the head position.
We define the position $\f = \front(\xi)$,  called the \df{front}, by
     \begin{align*}
         \front(\xi) = \xi.\pos + \ZigDepth\cdot d.
     \end{align*}
This is the farthest position to which the head has
advanced before starting a new backward zig.

Let $\delta=\Drift$.
Recall the definition of the \df{transfer sweep}
$\TransferSw(\delta)$ in~\eqref{eq:TransferSw},
if $\delta \ne 0$.
There is no transfer sweep if $\delta = 0$.
We require:
    \begin{flushdescription}
    \item[Colonies]
        The non-blank cells of the tape form a single segment,
        subdivided into colonies, starting from the \df{base}
        defined by counting back from $\Addr$ (this is not
        necessarily the origin of the tape).
        The leftmost colony and rightmost colony may be only
        partially filled.

        The colony starting at the base is called the
        \df{base colony}.
        There is also an \df{extended base colony} $X$:
        this is obtained by extending the base colony in the direction $\delta$,
        provided $\Sweep \ge \TransferStart$
        (defined in~\eqref{eq:TransferStart}).

        The front $\front(\xi)$ is always in the extended base colony.
        The drift of nonempty outer cells points towards
        the base colony.

    \item[Workspace]

        The non-outer cells form a single interval
        called \df{workspace}, with the following properties:

        \begin{bullets}
            \item For $\Sweep < \TransferSw(\delta)$,
                  it is equal to the base colony.
            \item In case of $\Sweep = \TransferSw(\delta)$,
                  it is the smallest interval including
                  the base colony and the cell adjacent to
                  $\front(\xi)$ on the side of the base colony.
            \item If $\TransferSw(\delta) < \Sweep < \Last(\delta)$,
                  then it is the extended base colony.
            \item When $\Sweep = \Last(\delta)$, it is
                  the smallest interval including the future
                  base colony and $\front(\xi)$ (it is shrinking onto the future
                  base colony).
         \end{bullets}
            The field $\cAddr$ varies continuously over
                  the workspace in all these cases, except
                  possibly $\Sweep=1$.

        \item[Sweep] For $1 \le \cSweep \le \Last(\delta)$,
            we have $\cSweep(x) = \Sweep$
            in all cells $x$ behind $\front(\xi)$ in the workspace.
            For $1<\cSweep$, we have $\cSweep(x)=\Sweep-1$
            in all cells $x$ ahead of $\front(\xi)$ (inclusive)
            in the workspace.

        \item[Addresses]
            Consider addresses $\cAddr$ in the workspace.
            Except for $\Sweep=1$, they increase continuously.

            In the first sweep, the address track $\cAddr$ is
            either $\lint{-Q}{0}$ or $\lint{Q}{2Q}$, but reduced
            modulo $Q$ on the segment
               $\lint{0}{\front(\xi)}$.

        \item[Drift]
         If $\Sweep \ge \TransferStart$ or $\Sweep=1$ then $\cDrift$
         is constant on the workspace.

        \item[Simulated content]
          The $\Info$ and $\State$ tracks contain valid codewords as
          defined in Section~\ref{subsec:coding}.

        \item[Normality] All cells are unmarked, that is $\cRec.\Core = 0$ throughout
           (see the definition of marking after~\eqref{eq:cRecCore}).

  \end{flushdescription}
  \end{definition}

The following observation comes directly from the definition of health.

\begin{lemma}\label{lem:workspace}
    In a healthy configuration, a cell is either under
    the head, or is in the workspace, or is an outer cell.
\end{lemma}

\begin{definition}[Local configuration, replacement]
  A \df{local configuration on} a (finite or infinite)
  interval $I$ is given by values assigned to the cells
  of $I$, along with the following information: whether
  the head is to the left of, to the right of or inside
  $I$, and if it is inside, on which cell, and what is
  the state.

  If $I'$ is a subinterval of $I$, then a local configuration
  $\xi$ on $I$ clearly gives rise to a local configuration
  $\xi(I')$ on $I'$ as well, called its
  \df{subconfiguration}: If the head of $\xi$ was in $I$
  and it was for example to the left of $I'$, then now
  $\xi(I')$ just says that it is to the left, without
  specifying position and state.

  Let $\xi$ be a configuration and $\zeta(I)$ a local
  configuration that contains the head if and only if
  $\xi(I)$ contains the head.
  Then the configuration $\xi | \zeta(I)$ is obtained by
  replacing $\xi$ with $\zeta$ over the interval $I$,
  further if $\xi$ contains the head then also replacing
  $\xi.\pos$ with $\zeta.\pos$ and $\xi.\state$ with
  $\zeta.\state$.
\end{definition}

\begin{definition}[Coordination]\label{def:state-coord}
  The state is called \df{coordinated} with the content
  of the observed cell if it is possible for them to be
  in some healthy configuration.
\end{definition}

Of course, it would be possible to give a finite table describing the coordination
conditions.
But we just point out some consequences of coordination we will use later:

\begin{lemma}[Coordination]\label{lem:coord}
Each $\Core=(\Addr,\Sweep,\Drift)$ value determines uniquely the
$\cCore$ value of the cell it is coordinated with.

In the reverse direction, the relation is less strict:
each $(\cAddr,\cSweep)$ pair determines uniquely the $\Addr$
that can be coordinated with it, and requires
$\Sweep\in\{\cSweep,\cSweep+1\}$, with the following exception:
when $\cAddr$ is within $4\beta$ of a colony end.

\end{lemma}
\begin{proof}
  The exception comes from the fact that there are two ways for the head to step
  onto cells
  of a neighbor colony: either during the transfer sweep, or at times when the
  head makes a turn at the end of a sweep, and after moving forward
  $\Z-4\beta$ steps, zigs back $\Z$ steps, thereby reaching $4\beta$ steps into
  the neighbor.
\end{proof}

To describe the self-correction process, we need to characterize the kind of
configurations that can be found during it.
We cannot hope to restore health
in all islands created by faults, in a very short
time after the faults occurred.
Indeed, as seen from Figure~\ref{fig:smart},
it may happen that a burst creates an island, but leaves
it with a state of the head that will not require it to
zig back anymore.
Moreover, this may happen in the last sweep of a
work period, while moving the base, say, to the left:
so the island created this way will be seen next, if ever, only if the
simulated computation transfers the base right again.

The following definition classifies the kinds of alteration that noise can
bring to a healthy configuration.
Informally, in islands, the structure may have been damaged, while
in stains, only the $\cInfo$ and $\cState$ tracks could be.
The distress area is where structure is currently being restored.
Recall that the $\Core=(\Addr,\Sweep,\Drift)$,
$\cCore$ and $\cRec.\Core$ tracks were introduced
in~\eqref{eq:Core} and~\eqref{eq:cRecCore}.

\begin{figure}[!ht]
\centering
\begin{tabular}{cc}
\includegraphics[width=2.5in]{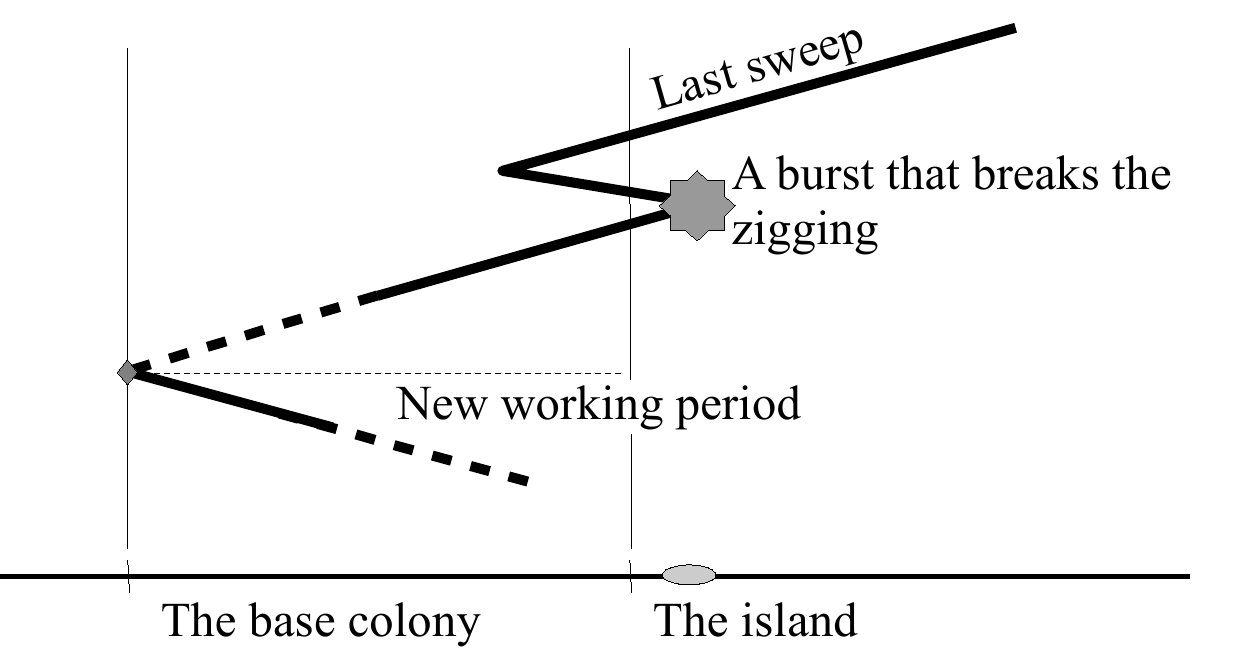} &
\includegraphics[width=2.5in]{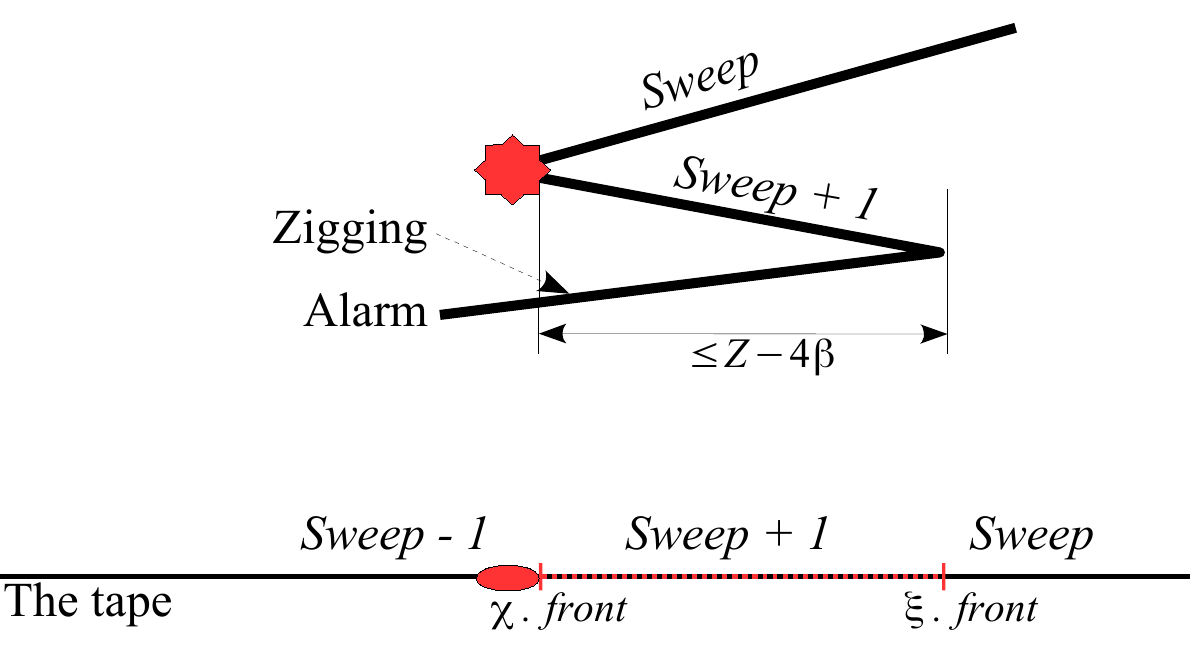} \\
(a) & (b)
\end{tabular}
\caption{(a) A burst during the last visit of the colony, at the bottom of a zig.
It puts the state into normal mode,
with appropriate values of $\ZigDepth$ and $\ZigDir$.
This leaves the created island ``undetected'' until the head returns to the colony
(b) A burst switches the sweep value causing the head to move forward and leaving
an island and a part of the tape without incremented sweep number.
\label{fig:smart}}
\end{figure}
\begin{definition}[Annotated configuration]\label{def:annotated-config}
    An \df{annotated configuration} is a quadruple
        \begin{align*}
            (\xi,\chi,\cI,\cS, \D),
        \end{align*}
    with the following meaning:
    $\xi$ is a configuration,
    $\chi$ is a healthy configuration,
    $\cI$ is a set of intervals of cells
    called \df{islands}, further $\cS \supset \cI$
    is a set of intervals of cells called \df{stains},
    $\D$ is an interval containing
    the head called the \df{distress area}.

    The distress area contains any island containing the head.

    Islands and stains are of size $\leq \beta$.
    The distress area has size $\leq 3\E$.

    We can obtain $\chi$ from $\xi$ by changing
    \begin{bullets}
         \item the $\cCore$ and $\cRec.\Core$ tracks
               in the islands and possibly
               additional $\le Z-3\beta$ cells within $D$;
         \item the $\cInfo$ and $\cState$ track in the stains;
         \item the state, the $\cRec.\Core$ track in $D$,
               and the head position inside $D$.
    \end{bullets}

    We say that an interval $W$ is the \df{workspace} of the annotated
    configuration $\cA$ if it is the workspace of $\chi$.

    The following additional properties are required:
    \begin{flushdescription}
        \item[Islands] 
            At most one island intersects the workspace.
            There are at most 2 islands in each colony that do not
            intersect the workspace.
            If there is more than one, then one is within
            distance
            $\E + 5\beta$   
            from the colony boundary
            towards the base colony.

        \item[Stains]
            In the base colony, either all stains but one are
            within a distance $\E+\beta$ to the left colony boundary,
            or all but one are within a distance $\E+\beta$
            to the right colony boundary.
            In all other colonies, all stains but one are
            within distance $\E + \beta$ of the boundary towards the
            base colony.
        \item[Distress]
            If $D$ is empty then the mode is normal.

    \end{flushdescription}

    We say that a cell is \df{free} in an annotated configuration
    when it is not in any island or $D$.
    The head is \df{free} when $D$ is empty.
    An annotated configuration is \df{centrally consistent}
    if the workspace is free.
\end{definition}

\begin{definition}[Admissible configuration]
\label{def:admissible-config}
    A configuration $\xi$ is \df{admissible} if there
    is an annotated configuration
      $(\xi,\chi,\cI,\cS, \D)$.
    In this case, we say that $\chi$ is a healthy
    configuration \df{satisfying} $\xi$.
    Any change to an admissible configuration is called
    \df{admissible}, if the resulting configuration
    is also admissible.
\end{definition}

The following key lemma shows that an admissible
configuration can be locally corrected.
Recall outer cells from Definition~\ref{def:outer-cells}.

\begin{lemma}[Correction]\label{lem:correction}
    Consider an annotated configuration
     \begin{align*}
             (\xi, \chi, \cI, \cS, \D),
     \end{align*}
    and an interval $\R = \lint{a}{b}$ of length $2\E$, further
    \begin{align*}
        \R_{i}^{j}=\lint{a + 0.1 i\E}{b - 0.1 j\E}\txt{ for } i,j=0,2,4.
    \end{align*}

    Assume that either in the left half or the right half of $\R$,
    at least $\E - 3\beta$ cells of $\xi(\R)$ are nonempty.
    Then it is possible to compute from $\xi.\cCore(\R)$ an interval
       $\hat \R \in \set{\R, \R_{0}^{4}, \R_{4}^{0}}$,
    a local configuration
       $\zeta = \zeta(\hat \R)$
    with no empty cells, such that
        $\chi|\zeta(\hat \R)$
    is healthy, and the following holds:

    \begin{enumerate}[\upshape (a)]
     \item\label{i:correction.inside}
        If $\chi.\pos \in \R_{2}^{2}$ then $\hat \R = \R$, $\zeta.\pos \in \R$,
        and $\zeta.\ZigDepth = 0$.

        If $\chi.\pos < a + 0.2 \E$ then $\hat \R = \R_{4}^{0}$,
        and $\zeta.\pos$ is to the left of $\hat \R$.
        Similarly, if $\chi.\pos \ge b - 0.2\E$ then $\hat \R = \R_{0}^{4}$,
        and $\zeta.\pos$  is to the right of $\hat \R$.

     \item
        The states of nonempty cells of $\xi$ can differ from the corresponding cells
        of $\zeta$ only in the islands, and in at most 
        $\Z-3\beta$ additional positions in an interval in $\D$ containing $\chi.\front$.

     \item
        The computation of $\zeta$ can be carried out by the machine $M_{1}$
        (relying only on $\xi$ and $\R$), using a
        constant number (independent of $\beta$, $Q$)
        of passes over $\R$, and a constant number
        of fields containing values of size $\le Q$.
    \end{enumerate}
\end{lemma}

\begin{Proof}
     For any interval $I$, let $\alpha(I)$
     denote the majority value of $\xi.\cAddr(x)-(x-a)$ over $I$, further
     $\sigma(I)$ and $\delta(I)$ the majority value
     of $\xi.\cSweep(x)$, and $\xi.\cDrift(x)$ over $I$.
     Let $m_{\alpha}(I)$, and so on, denote the multiplicity of $\alpha(I)$, and so
     on, over interval $I$.

    We now outline a procedure that finds $\hat \R$ and $\zeta$.
    Even if the reasoning below refers to the
    healthy configuration $\chi$ occasionally, the computation
    only relies  on the configuration $\xi$.
    Whenever we write \df{plurality}, we mean a value with
    multiplicity larger than $1/3$ of the total.
    Empty cells are not counted in the total, and do not contribute to the counts.

 \begin{step+}{step:correction.cases}
  We have the following, not mutually exclusive
  possibilities, that can be
  checked:
     \begin{bullets}
         \item
            All but $3\beta$ cells of the left/right half of $\R$ are outer cells.
         \item
            All but $3\beta$ cells of the left/right half of $\R$ are workspace cells.
     \end{bullets}
    \end{step+}
    \begin{pproof}
        This follows from the fact that in
        the healthy configuration $\chi$, the workspace is
        surrounded by outer cells.
    \end{pproof} 

    \begin{step+}{step:correction.outside}
        Assume that at least $1.7\E$ cells of $\R$ are
        left outer cells of $\xi$, or at least $1.7\E$
        are right outer cells.
    \end{step+}
    \begin{prooofi}
        Without loss of generality, assume that at least
        $1.7\E$ cells in $\R$ are left outer cells: set $\hat\R \gets \R_{0}^{4}$.
        The value $\sigma\lint{a}{b-\E/2}$ is necessarily $\Last(1)$.
        Let $\alpha = \alpha\lint{a}{b-\E/2}$.
        Setting
        $\zeta.\cAddr(x) = \alpha+x-a$,
        $\zeta.\cSweep = \Last(1)$, and
        $\zeta.\cDrift = -1$ defines
        $\zeta.\cCore(x)$ accordingly for all $x$ in $\hat \R$ (not leaving
        empty cells).
    \end{prooofi} 

Assume that the above test fails: then given that
$\R$ intersects  at most
3 islands, we can assume that at least
$0.3\E - 3\beta$ cells of $\R$ belong to
the workspace of the healthy configuration $\chi$.

\begin{step+}{step:correction.inside}
  Suppose that $\xi$ has at most $3 \beta$ outer cells in $\R$.
\end{step+}

\begin{prooofi}
       Then the non-workspace cells of
        \begin{align*}
            \R^{-} = \lint{a + 3\beta}{b - 3\beta}
        \end{align*}
       are all island cells, since the non-island non-workspace
       cells of $\R$ must all  be at the ends.

       Compute $\sigma(R^{-})$, and assume without loss of generality
       $\dir(\sigma) = 1$ (the right sweep).

       We claim
             \begin{align*}
                   \chi.\front - \Z
                      \le
                         a^{+} + m_{\sigma}
                      \le
                        \chi.\front+\Z,
             \end{align*}
    where $a^{+} = a + 3 \beta$.
    Indeed, in the healthy configuration $\chi$, the right-sweeping cells inside
    $\R^{-}$ form an interval on the left of $\chi.\front$.
    By the definition of annotation, $m_{\sigma}$ could differ from the size of this
    interval only due to island cells, and possibly an interval of size
    $\le \Z - 3\beta$ containing $\chi.\front$.

\begin{step+}{step:correction.inside.addr}
  Compute the addresses and sweep values.
\end{step+}
\begin{prooofi}
  Recall that we assumed $\dir(\sigma)=1$.
  First we compute the candidate address and
  sweep values in  $\lint{a}{a^{+}+m_{\sigma}}$.

  Note that $\xi.\cAddr(x)-(x-a)$ should be constant as $x$
  runs on all non-island workspace cells of $\R^{-}$
  with the above plurality value of $\sigma$.
  Therefore it has some majority value $\alpha$.

  For cells $x$ in $\lint{a}{a^{+}+m_{\sigma}}$, let
  $\zeta.\cSweep(x)\gets\sigma$,
  $\zeta.\cAddr(x)\gets\alpha+x-a$.
  This can change only island cells or shift the front to the left by
  a number of cells equal to
  the number of island cells encountered in this interval.

  If $a^{+}+m_{\sigma}\ge b-0.3\E$, then
     set $\hat \R\gets \R_{0}^{4}$.

  Assume now $a^{+}+m_{\sigma}<b-0.3\E$, and  set $\hat\R \gets \R$.

  Now we compute the candidate address and
  sweep values in $\lint{a^{+}+m_{\sigma}}{b^{-}}$.
  Let $\sigma'$ be the majority value of
  $\xi.\cSweep$ in $\lint{a^{+}+m_{\sigma}}{b^{-}}$, where $b^{-}=b-3\beta$
  (the majority exists, due to admissibility).
  Note that $\xi.\cAddr(x)-(x-a)$ and $\xi.\cSweep(x)$
  should be constant for almost all $x$ in $\lint{a^{+}+m_{\sigma}}{b^{-}}$.
  For $x$ in $\lint{a^{+}+m_{\sigma}}{b}$, set
  $\zeta.\cSweep(x) \gets \sigma'$, $\zeta.\cAddr(x) \gets \alpha'+x-a$.
  This again can only change island cells or possibly some cells due to
  the left shift of the front.

  In this way, the total number of cell changes
  is at most $3\beta+\Z$: at most $3\beta$
  in the islands and $3\beta+(\Z-3\beta)$ due to the shift of the front.
\end{prooofi} 

    \begin{step+}{step:correction.inside.rest}
      Compute the remainder of $\zeta$.
    \end{step+}
    \begin{prooofi}
      Assume first $|\sigma'-\sigma|=1$, that is
      two consecutive sweep values within a work period.

      If $\sigma'<\sigma$, then set
      $\zeta.\front\gets a^{+}+m_{\sigma}$,
      $\zeta.\Sweep\gets\sigma$;
      otherwise
            $\zeta.\front\gets a^{+}+m_{\sigma}-1$,
            $\zeta.\Sweep\gets\sigma'$.
      If $\min(\sigma,\sigma')\ge\TransferStart$ then
      all over $\hat \R$, set
      the $\zeta.\cDrift$ values to the majority
      of the $\xi.\cDrift$ values over $\R$.

      Assume now that $\sigma,\sigma'$ are the two
      values corresponding
      to the transition to a new work period.

      By assumption $\sigma=1$; we set
      $\zeta.\front\gets a^{+}+m_{\sigma}$.

      The value $\chi.\cDrift$ has a constant value $\delta$ on $\R$.
      We can determine it using majority of
      $\xi.\cDrift$ over $R^{-}$, and replace it all over $\hat\R$.
     \end{prooofi} 

\end{prooofi} 

    Assume that the test~\ref{step:correction.inside} also fails:
    then $\R$ intersects the workspace without being
    essentially contained
    in it, or essentially disjoint from it.
    From the four possibilities of
    part~\ref{step:correction.cases} above,
    now only these remained: one half of $R$ is
    essentially covered by workspace cells
    and the other one is not, or one half of $R$
    is essentially covered by outer
    cells and the other one is not.
    We can therefore make, without loss of generality,
    the following assumption:

    \begin{step+}{step:correction.fixed}
      Assume that the left half of $R$ is not
      covered essentially (that is to within
      $3\beta$) by outer cells.
      Also, either the left half is covered
      essentially by workspace cells, or the
      right half is covered essentially by outer cells.
    \end{step+}
    \begin{prooofi}
     Let $m$ be the number of workspace cells of $\xi$ in $\R$.
     Then the intersection of the workspace with
     $R$ must agree within $3\beta+\Z$ with $\lint{a}{a+m}$,
     just as above in
     part~\ref{step:correction.inside.addr}.
     Now we can carry out the computations of
     part~\ref{step:correction.inside}
     in the interval $\lint{a}{a+m}$ in place of $R^{-}$.
     Since $m \ge 0.3\R - \Z$, there will be still
     sufficient cells left in this
     interval for the correct computation of the majorities.

    \end{prooofi} 

        It is straightforward to check that the conditions
        of the lemma guarantee that the construction of $\zeta$ has
        the properties claimed in the lemma.

\end{Proof}

Assuming that the conditions of Lemma~\ref{lem:correction}
hold, it is clearly possible to compute a constant upper bound on
the number of sweeps of the domain
$\R$ needed for the machine $M_{1}$ to
perform the calculations, resulting in a bound $O(\beta)$
on the total number of steps used.

\begin{definition}[Correction data]\label{def:correction-data}
  The following information
    $\Delta = (s,\alpha,\alpha',\sigma,\sigma',\delta,\f)$
  incorporates all the data defining the corrected healthy
  local configuration $\zeta(\hat \R)$, provided $\R$ is given:
  \begin{bullets}
    \item
      $s\in\{-1,0,1\}$ says which of the three
      alternative values of $R_0^4$, $R$, $R_4^0$ does $\hat \R$ have.
      (In this case, $\hat R = [a + \max(0,s) 0.4\E, b + \min(0,s)0.4\E]$.)
    \item
      $\alpha,\alpha',\sigma,\sigma'$ help computing
      the address and sweep values as seen in the proof
      of the Correction lemma.
    \item
      $\delta$ is the $\cDrift$ value shared by all
      elements of the workspace inside $\hat \R$
      in case $\sigma\ge\TransferStart$ or $\sigma = 1$.
    \item
      $\f = \front(\zeta) - a$ in case $\hat\R = \R$.
  \end{bullets}
\end{definition}

\subsection{Recovery procedure}\label{subsec:recovery}

Starting from a point $x$, the recovery procedure opens an interval
 \begin{align*}
   \R = z_{1} + [-\E,\E), \quad \text{with } z_{0} = z_{1} - \E,
 \end{align*}
to which it applies the algorithm of the proof of the
Correction Lemma~\ref{lem:correction}.
There is a point in which the Correction Lemma did not
specify all the changes:
    when $\hat \R=\R_{0}^{4}$ then it only said that the
    head should go to the right of $\hat \R$, not the exact place where it should go.
    In this case, the procedure will put the head at $z_{1} + \E$,
    that is all the way to the right edge of $\R$.
    Similarly, if $\hat \R= \R_{4}^{0}$ then the head goes
    to $z_{1} - \E$, that is on the left edge of $\R$.
    In both cases we set $\ZigDepth=0$.

The following example shows the need for a careful choice of the recovery
interval.

\begin{example}
    [Motivation for aligned recovery intervals]\label{ex:aligned}
    Denote $C(b)$ a colony with starting point $b$.
    Consider the following scenario.
    During the rightward transferring sweep to colony
    $C(b)$,  while within distance $\E - 4\beta$ of the right boundary $b+Q$,
    the head hits an island, calling alarm.
    The recovery procedure opens a recovery interval and proceeds to work on
    it.
    Now, while the head is on the right boundary of this interval, a burst occurs.
    As a result of this burst, nothing changes inside the recovery interval,
    or in the head position or the state, but an
    island $I$ is created on the right, outside of the recovery interval.
    Assume that the computation  from now on continues to the left of $b+Q$.
    In some much later work period, at the last sweep before moving left from
    colony $C(b)$, a burst leaves an island within distance
    $\E - 5\beta$ from $b + Q$.
    Then, in some much later work period, during the transferring sweep to
    $C(b)$, the head hits this new island and the recovery starts.
    Now we repeat the same scenario as above, creating an island $I - \beta$ which
    will stay there.
    If we continue with this adversarial way of
    putting islands, the entire interval
    $b + 4\beta + \lint{0}{\E +\beta}$ can be covered by islands.
    Then, much later, in a transfer to colony $C(b+Q)$,
    the algorithm of the Correction Lemma~\ref{lem:correction}
    may be defeated.
\end{example}

To prevent such scenarios, the recovery procedure will try to ensure that the
recovery interval have the following special property.

\begin{definition}
    An interval is called \df{aligned} if its endpoints are divisible by $\E$.
We require
\begin{align}\label{eq:requirement-on-Q}
  \E \mid Q.
\end{align}
\end{definition}

For controlling the details,
the procedure uses the field $\Rec.\Addr$ to measure
the distance from point $a$, and a field $\Rec.\Sweep$
to measure the progress, just as in the main program.
There are corresponding $\cRec.\Addr$ and $\cRec.\Sweep$
fields in the cells.
According to the values of $\Rec.\Sweep$, we distinguish \df{stages},
and introduce the \emph{pseudofield} $\Stage$ (it is just a function of
$\Rec.\Sweep$), with values
 \begin{align*}
   \Stage\in\{\Marking,\Planning_{i}\; (i=1,2), \Mopping_{i} (i=1,2)\}.
 \end{align*}
The process makes use of a number of rules:
$\Alarm$, $\Mark$, $\Plan(i)$, $\Mop(i)$ for $i=1,2$.
Whenever we say that a rule ``checks'' something,
it is understood that if the check fails, alarm is called.
In all rules but in $\Mop$, wherever the head steps,
it walks on marked cells or it marks them,
that is it sets $\cRec.\Core \ne 0$.
The rule $\Mop$ is devoted to unmarking.
Zigging will be performed using the fields
             \begin{align*}
                   \Rec.\ZigDepth, \Rec.\ZigDir,
             \end{align*}
and the constant parameter
             \begin{align}\label{eq:Rec.Z}
               \Rec.\Z = 11\beta.
             \end{align}
However, even while zigging, the head stays strictly within the recovery
interval.

The following rule is going to run simultaneously
through all the rest of the recovery procedure.
              \begin{bullets}
                  \item Check if $c.\Rec.\Addr = \Rec.\Addr + d$, where
                        $d=\pm1$ is the direction of the sweep.
                  \item If not zigging, check if $c.\Rec.\Sweep = \Rec.\Sweep -1$.
                        If zigging, check if $c.\Rec.\Sweep = \Rec.\Sweep$.
                  \item Update the field $\Rec.\Addr$ in every move,
                        increasing or decreasing it as we move left or right.
              \end{bullets}

\begin{enumerate}[1.]
    \item\label{i:recovery.alarm}
        The rule $\Alarm$ sets
            $\Mode\gets\Recovering$,
            $\Stage\gets\Marking$.

    \item\label{i:recovery.mark}
        Rule $\Mark$ locates and marks the recovery area with
            $\cRec.\Core\gets 1$, and moves to $z_{0}$.
        (The meaning of the value 1 is that the cell is marked,
        but we did not assign any useful $\Core$ values to it yet.)
        It alarms if any of the cells along the way that it expects to be
        already marked is not.

        In order to mark $\R$, the head moves in a zigging way,
            similarly to what is done in the main simulation,
            as described in point~\ref{i:fields.zigging} of
            Section~\ref{sec:structure},
            except that we do not go outside the interval $\R$.
            Zigging makes sure not to mark too many cells in one sweep
            or without checking that they are marked consistently with
            what was marked before.

        After determining the interval $\R$ from examining a segment of
        $14\beta$ cells, the rule marks one half of this interval, then
        passes over the marked half to mark the other half.
        Here are the details.

        Let $\lint{x_{0}}{x_{1}}$ be the aligned interval
        of length $\E$ containing the cell $x$ where
        alarm was called.

        \begin{enumerate}[i.]
         \item\label{i:recovery.find-alignment}
            This part starts from a cell $x$ (where the alarm was called),
            and ends on cell $x+7\beta$.
            In its sweep 1, moving left, it remembers the majority
            of $\cAddr(y)-(y-x) \bmod Q$
            for $y\in x+\lint{-7\beta}{0}$
            as a candidate modulo $Q$ address $\lambda_{-1}$ for $x$.
            If there is no such majority, the value is undefined.
            It also computes a majority sweep value $\sigma_{-1}$ if a majority exists.
            Now, the machine turns right and while passing over $\lint{x}{x+7\beta}$ it
            computes $\lambda_{1}$ and $\sigma_{1}$ similarly.
            Admissibility implies that if both $\lambda_{j}$ are defined then they are
            equal.
            Moreover, if both are defined then at least one of the $\sigma_{j}$ is defined.

            From these values, we will compute a candidate mod $Q$ address $\lambda$ and a
            candidate direction $\delta$ as follows.

            \begin{enumerate}[(1)]

            \item
                If one of the pairs $(\lambda_{j}, \sigma_{j})$
                is defined and the other one is not, then
                $\lambda\gets\lambda_{j}$, $\delta\gets -j$
                (direction is towards the undefined pair).
                Otherwise let $\lambda$ be the common value of the $\lambda_{j}$.

            \item
                If $\sigma_{j'}\le\sigma_{j}\le\sigma_{j'}+1$ or
                $\sigma_{j}<\sigma_{j'}$ then $\delta\gets(-1)^{\sigma_{j}+1}$,
                that is $\delta$ is the direction of the current sweep as defined
                in~\eqref{eq:sweep-dir}.
            \end{enumerate}

        From $\lambda$ we can compute the values $x_{0},x_{1}$.
        Now we determine $z_{1}$ as follows:
        If $|x-x_{j}|<0.2\E$ for some $j$ then let $z_{1}=x_{j}$.
        Otherwise, let $z_{1}=x_{j}$ for the $x_{j}$ with
        $\sign(x_{j}-x)=\delta$.

        The rule achieves the following conditions.
         \begin{enumerate}[a.]
             \item The point $x$ is in $\R$, at least $0.2\E$
                 away from its boundary.
             \item\label{i:recovery.distance-from-front}
              If $|x - \front(\chi)| < 0.1\E$, then
              $\R$ reaches less than $1.3\E$ backwards
              from $\front(\chi)$.

            Indeed, without loss of generality assume that the
            direction of the sweep is 1.
            From $x - z_1 \leq 0.2\E$,
            we obtain
            \begin{align}\label{eq:eq}
               x - z_1 + \E  \leq 1.2\E.
            \end{align}
            From our assumption and~\eqref{eq:Expansion}
            we have $x \geq \chi.\front - 0.1\E$.
            Applying it to~\eqref{eq:eq}
            yields $\chi.\front - (z_1 -\E) < 1.3\E$.
           \end{enumerate}

        At the end of this rule, being in a cell $y$, we set the field
         \begin{align*}
             \Rec.\Addr=y-z_{0}.
         \end{align*}
    \end{enumerate}

    \item
        The rule $\RangeCheck$ checks that all cells of $\R$ are marked.

    \item
        Rule $\Calc$ carries out, over interval $\R$,
        the algorithm of the Correction Lemma~\ref{lem:correction}
        to determine the interval $\hat \R$ and the local configuration $\zeta(\hat \R)$.
        If none of the cases apply in the algorithm described in
        the proof, the rule calls alarm.
        It remembers the computation result in a field $\Delta$ as given in
        Definition~\ref{def:correction-data}.

    \item Stages $\Planning_{1}$ and $\Planning_{2}$ follow each other.
        Stage $\Planning_{i}$ calls rule $\Plan(i)$.

        $\Plan(i)$ calls $\RangeCheck$ and then $\Calc$.
        In case $i=1$, it writes the resulting $\zeta.\cCore$ values
        on the
        $\cRec.\Core$ track of $\hat \R$, and $\cRec.\Core\gets 1$
        into $\R\setminus \hat \R$.
        In case $i=2$, it just checks whether the result is equal to
        the existing values of $\cRec.\Core$.

    \item Stages $\Mopping_{1}$ and $\Mopping_{2}$ also follow each other.
        Rule $\Mop(1)$ unmarks the cells over $\R$,
        setting
           $\cCore\gets\cRec.\Core$
        at the same time,
        if $\cRec.\Core\not\in\{0,1\}$.
        It relies on the field $\Rec.\Addr$ measuring
        the distance $x-z_{1}$
        of the current cell $x$ from $z_{1}$, and also
        on part $\f$ of the data $\Delta$ introduced
        in Definition~\ref{def:correction-data},
        (and computed in each stage $\Planning_{i}$).

        If $\hat \R=\R_{0}^{4}$ then Rule $\Mop(1)$ moves
        from the left end of $\R$ to the right end while
        unmarking, and stays there.
        If it is $\hat \R=\R_{4}^{0}$ then it moves from
        the right end to the left end while unmarking.
        Otherwise, it first moves to the end of $\R$ in direction
          $-\dir(\zeta.\Sweep)$
        (that is backward from the sweep direction
        from $\zeta$), and then erases the marks up to
        position $\zeta.\front$.
        Then Rule $\Mop(2)$ follows, which is similar,
        but works from the other direction, ending up at $\zeta.\front$
        with no marked cells.

        Zigging is used during the mopping stage just
        as during the marking stage.
\end{enumerate}

\begin{remark}[More on alignment]\label{rem:explanation-R}
   One solution for the problem presented in
   Example~\ref{ex:aligned} would be to zig also outside of the recovery
   interval during the mopping phase.
   However, this would open the door for the errors to influence the recovery
   interval in a sliding manner in yet another, but similar way.
   Alignment snaps the interval $\R$ always to center on the colony boundary,
   preventing a sliding contamination with islands.
   (Stains can still be created in the neighbor colony, but as we will see later in
   Lemma~\ref{lem:recovery}, they stay within $\E + \beta$ cells from the
   colony boundary.)
\end{remark}

It is easy to check that the recovery procedure uses only a constant number of
sweeps, for a total number of steps
\begin{align}\label{eq:K-Recovery}
 \K_{\R} = O(\beta).
\end{align}

\section{Proof of the Main Theorem}\label{sec:proof}

It is useful to spell out the kind of
simulation that machine $M_{1}$ performs.

\begin{definition}\label{def:noisy_simulation}
    A computation history in the sense of  Definition~\ref{def:configs,histories}
    is a \df{$(\beta,V)$-noisy trajectory}, if faults in it are confined to bursts
    of size $\le \beta$ separated by time intervals of size $\ge V$.

    A pair of mappings $(\varphi_{*},\Phi^{*})$
    in the sense of Definition~\ref{def:simulation}
    is a \df{$(\beta,V)$-tolerant simulation}
    of Turing machine $M_{2}$ by Turing machine
    $M_{1}$ if for every string $x \in \Sigma_{2}^{*}$, every
    $(\beta,V)$-noisy trajectory $\eta$ of $M_{1}$
    whose initial configuration is $\varphi_{*}(x)$,
    the history $\Phi^{*}(\eta)$ is a trajectory of $M_{2}$.
\end{definition}

The proof of the main theorem will show, as a side result,
that our simulation is a $(\beta,V)$-tolerant
simulation of $M_{2}$ by $M_{1}$.
We assume that the output of $M_2$ is 0 or 1 written in cell 0.
It is time to define more precisely the concepts
connected with recovery.

\subsection{Annotated history}

Let us analyze the kind of histories that are possible with
sparse bursts of faults.
Recall the definition of (possibly centrally consistent)
annotated configuration in Definition~\ref{def:annotated-config}.

\begin{definition}[Annotated history]\label{def:annotated-hist}
    An \df{annotated history} is a sequence of annotated configurations
    if its sequence of underlying configurations is a $(\beta,V)$-trajectory,
    and it satisfies some additional requirements given below.

    If the head is in a free cell, in normal mode, then
    the time (and the configuration) will be called \df{distress-free}.
    If the annotated configuration at a certain time
    is centrally consistent, then we call
    that time \df{centrally consistent}.
    A time that is not distress-free and was preceded by
    a distress-free time
    will be called a \df{distress event}.

    Consider a time interval $\lint{t}{t+u}$ starting
    with a distress event and ending with the head
    becoming free again.
    It is called a \df{relief event} of \df{duration} $u$ if
    the only possible island that remains from the distress area
    is due to some burst that occurred at a time intersecting
    $\lint{t}{t+u}$.
    Moreover, if such an island
    exists, then the sweep direction from before the
    distress event is preserved, except when the island is
    outside the extended base colony---then it will be reversed.

    The \df{extent} of a relief event is the maximum size
    interval covering the
    distress area during the distress.

    Recall the definition of the parameter $\K_\R$ in~\eqref{eq:K-Recovery}.
    The additional requirements for annotated history are:

    \begin{enumerate}[(a)]

        \item \label{i:annotated-hist.progress}
            Islands are only created by noise.
            Stains and the distress area start out as islands.

        \item\label{i:annotated-hist.distress}
            Each distress event
            is followed immediately by a relief event,
            of duration $\le 3\K_\R$ and extent $\le 3\E$.

        \item\label{i:annotated-hist.stain}
            If a distress-free configuration has  $\Sweep \ge \TransferStart$,
            then the base colony contains no stains from earlier work periods.

   \end{enumerate}
\end{definition}

Lemma~\ref{lem:recovery} (Recovery) will be a crucial step
towards the proof of the main theorem.
Before spelling it out and proving it, we provide
some preparatory lemmas.

\subsection{Undisturbed recovery}

The idea of the proof of relief from damage is the following.
If alarm is called and the recovery process is allowed to complete,
then it carries out the needed correction, as guaranteed by the
Correction Lemma~\ref{lem:correction}.
Most complications are due to the fact that the state after a burst is
arbitrary.

When the mode is normal then zigging will make sure that
the effect is limited to near the island where the burst
happened: for example, the direction of a sweep
cannot be changed in the middle of the workspace,
since then zigging would notice this and call alarm.

But the mode after the burst can be the recovery mode,
with arbitrary values for all fields.
Moreover, a new burst may occur after an alarm, at an
arbitrary stage of the recovery.

In this section, we address the cases when two bad effects
do not combine:
either an alarm is called and completes without a new burst
intervening, or a burst occurs at a distress-free time.
Recall the definitions of the constants $\K_\R$
in~\eqref{eq:K-Recovery} and $\Z$ in~\eqref{eq:Z}.

\begin{lemma}[Undisturbed alarm]\label{lem:alarm}
    Suppose that in an annotated history,  alarm sounds at a time when the front
    of the healthy configuration $\chi$ is at a distance at most
    $2\Z$ from the head\footnote{
        The worst case occurs when the front is
        within $\Z-4\beta$ cells from a colony boundary
        and the head while zigging visits the
        neighboring colony where, within the
        first $4\beta$ steps a burst occurs and puts the state
        into marking (right locating branch).
        Alarm will be called closer to the front,
        and the distress area can grow by
        up to $(\Rec.\Z - 4\beta)+\beta$.},
    and the distress area does not stretch more
    than total size $2\Z$. 
    Suppose also that no burst occurs in the next $K_{\R}$ steps.
    Then the annotation of
    the history can be extended so that relief comes in fewer than
    $K_{\R}$ steps, while no more than $2\E$ cells are added to
    the distress area before it disappears.
\end{lemma}
\begin{proof}
    Assume that the conditions of the lemma hold.
    Let $x$ denote the position of the head at the
    moment when alarm is called. Let us follow the
    recovery procedure, to show how the relief is
    achieved.

    After the alarm, in the first two sweeps of the
    recovery procedure, interval $\lint{x-7\beta}{x+7\beta}$
    is created and then, an interval $\R$ is opened
    that extends the distress area.
    For the procedure to succeed, the
    condition of the Correction Lemma~\ref{lem:correction}
    must hold that in one half of $\R$ no more than
    $3\beta$ cells are empty.
    This is trivially true even when the alarm is
    called on the very first few steps of the
    simulation (since we have assumed that the address
    fields of the base colony and its two neighbors are
    nonempty).

    Recall the notation $\R = \lint{a}{b}$ and $\hat \R$
    of the Correction Lemma~\ref{lem:correction}.
    In its proof, we used $m_{\sigma}$ to denote the
    multiplicity of the plurality sweep $\sigma$
    within the interval $\lint{a + 3\beta}{b - 3\beta}$.
    Without loss of generality, assume that  the direction of $\sigma$ is 1.
    We will have
     \begin{align*}
        \chi.\front \le a + 1.3\E = b - 0.7\E,
     \end{align*}
    and therefore $\chi.\front$ is not to the right of $\hat\R$.
    Indeed, the assumptions of the lemma along with
    definitions of $\Z,\E,\Rec.\Z$ in~\eqref{eq:Z}, \eqref{eq:Expansion}
    and~\eqref{eq:Rec.Z} imply that $x$
    is not further than $0.04\E<0.1\E$ from $\chi.\front$.
    Now the claim follows from
    property~\ref{i:recovery.distance-from-front} of the
    recovery procedure.

    Furthermore, at least $m_{\sigma}$ right sweeping cells in $\R$
    will be on the left of $\chi.\front$.
    As a majority among not fewer than $\E-3\beta$ cells,
    $m_{\sigma} \ge (\E-3\beta)/2$.
    This shows that $\chi.\front$ is not to the
    left of $\hat\R$, hence $\hat \R = \R$.
    It follows that the recovery procedure erases the marks in the distress
    area, and rewrites all island cells in $\R$, allowing to erase the distress
    area and the islands to get relief within $\K_{\R}$ steps.
\end{proof}

\begin{lemma}[Burst]\label{lem:burst}
    Assume that the history has been annotated up to a time when
    a burst, creating an island $J_{0}$, occurs at a distress-free time.
    Then the burst is followed by a relief event of duration
     $\le \K_\R + \Z$
    and extent $\le 3\E$.
\end{lemma}

\begin{Proof}
    We consider various situations after the burst.
    Recall that we called an interval $\R$ of length
    $2\E$ \df{aligned} if in a
    healthy configuration satisfying the present one,
    its ends have addresses divisible by $\E$
(equivalently, if its end positions as absolute integers are divisible by $\E$).

Let $\chi$ denote the healthy configuration that is part of the annotation at
the time of the burst.
Since the burst occurs at a distress-free time, the head is within $\Z$ from
$\chi.\front$ when it happens.
In what follows, we will sometimes refer to $\chi.\front$ of this moment as just
the \df{front}.

    \begin{step+}{step:burst.normal}
        Assume first that the mode immediately after the burst is normal.
    \end{step+}
    \begin{prooofi}
        Without loss of generality, assume that the sweep was to the right.
        We start at some position $x$ that is either in island $J_{0}$ or next to it.
        Now the head zigs backward and forward
        $\Z$ steps (see~\eqref{eq:Z}), with respect to the sweep direction,
        between any two sequences of $\Z-4\beta$
        forward moving steps.
        In any of these, it may discover an incoordination
        and call alarm, in which case Lemma~\ref{lem:alarm}
        becomes applicable.

    \begin{step+}{step:burst.normal.same-sweep}
        Assume first that the burst does not change the sweep and address.
    \end{step+}
    \begin{prooofi}
      In this case, the head will continue its forward sweep, with just possibly
      changed zigging.
      While hitting elements of the island $J_{0}$ it may sense incoordination
      and call alarm.
      If this happens then Lemma~\ref{lem:alarm} applies, since we are at most
      $\Z+3\beta$ steps behind front, and at most $3\beta$ steps ahead it.

      Before the head manages to traverse $J_{0}$, it
      may hit another island causing an alarm.
      The point where this alarm can be called may be at most
      $\Z-3\beta$ steps ahead of the front, so Lemma~\ref{lem:alarm} applies
      again.

      In case of alarm, the recovery area will cover the island $J_{0}$, and if
      it was triggered by another island then that one, too.
      Any points of island $J_{0}$ traversed during the
      progress and zigging can be erased from the island, and after a complete
      cycle of zigging occurs the untraversed parts of $J_{0}$ may stay as an
      island.

      How can it happen that not the whole $J_{0}=:\lint{a}{a+\beta}$ is
      traversed?
      In this case, the next backward zig does not cross $J_{0}$, so it starts
      from $\ge a+\Z$.
      To get there we need $\chi.\front\ge a+\Z-(\Z-4\beta)=a+4\beta$ when we
      start.
     \end{prooofi}

    \begin{step+}{step:burst.normal.changed-sweep}
        Assume now that the burst changes $\Addr$ or $\Sweep$.
     \end{step+}
     \begin{prooofi}
       Lemma~\ref{lem:coord} says that unless $\cAddr$ is in a certain interval
       of length $4\beta$, the pair $(\cAddr,\cSweep)$ pair determines
       uniquely the $\Addr$ value coordinated with it.
       If the burst changes $\Addr$ then therefore this will be noticed as soon as the
       head leaves the island and possibly this interval, causing an alarm, so
       Lemma~\ref{lem:alarm} applies.

       Similarly, if $\Sweep$ changes by more than 1 then it will be noticed, as
       soon as the head leaves the island or the area of size $4\beta$
       mentioned.
       If it just changes by 1 then the head reverses direction, and the
       incoordination may not be immediately noticed when stepping off the
       island.
       But zigging will take us all the way across $J_{0}$ and
       therefore if alarm does not sound, $J_{0}$ can be erased (this can only
       happen if $J_{0}$ is at the end of the colony where the sweep would have
       changed anyway).
       Indeed, just as above, we can see that
       the only possibility that the next backward
       zig does not cross $J_{0}$ would be that
       the front is to the left of $a+\beta$ by at
       least $4\beta$.
       But this is impossible, since as the original sweep is to the right,
       the head was not right of the front when the island occurred.
     \end{prooofi}
\end{prooofi}

\begin{step+}{step:burst.recovering}
    Suppose that the mode after the burst is $\Recovering$.
\end{step+}
\begin{prooofi}
    In the recovery rule, as defined in Section~\ref{subsec:recovery},
    the head moves around in an aligned interval $\R$ of size $2\E$.
    The marked area is extended in stage $\Marking$, and shrunk in
    stages $\Mopping_{i}$.
    If the stage after the burst
    is $\Planning_{i}$, then alarm is called almost immediately
    (possibly passing through some island cells first),
    since we assumed a start from a distress-free configuration,
    in which by definition no non-island cells are marked.
    Then Lemma~\ref{lem:alarm} applies.

    \begin{step+}{step:burst.recovering.marking}
        Suppose that the stage after the burst is $\Marking$.
    \end{step+}
    \begin{prooofi}
        By its design, the marking rule marks new cells
        while also using a rule similar to $\Zigzag(d)$, but moving (and marking)
        at most $\Rec.\Z-4\beta$ cells while moving in one direction.
        Alarm is only called when an alignment problem is found,
        or non-marked cells are found where marked ones are expected.
        Therefore alarm can only occur within the
        first $2\Rec.\Z$ steps after a burst.
        Indeed, zigging checks alignment with the cells marked earlier.
        If alignment inconsistency is not found then it will not be found
        later either.

        It follows that in case of new alarm, Lemma~\ref{lem:alarm}
        applies, and the recovery reprocesses all cells marked after the burst.
    \end{prooofi}

    \begin{step+}{step:burst.recovering.mopping}
       Assume that after the burst a mopping stage is entered.
    \end{step+}
    \begin{prooofi}
        The mopping stages only erase marks, and apply
        $\cCore\gets\cRec.\Core$.
        Since we started in a distress-free configuration, we had
        $\cRec.\Core=0$ everywhere but in the islands.
        Marking will not change the $\cCore$ value anywhere else.
        It follows that within at most as many cells as the total
        length of islands possibly encountered, there is either an
        alarm due to not seeing marks, or return to normal mode.
        From there on, the analysis of part~\ref{step:burst.normal}
        applies.
    \end{prooofi} 
\end{prooofi} 
\end{Proof}

\subsection{Disturbed recovery}

We would say that recovery is disturbed when a
burst occurs during a recovery process started by an alarm.
Since bursts are rare, the alarm in question must have happened
then without a burst, which could occur
only on encountering some island $J_{1}$.
Since there was no recent burst (within the last $V$ steps),  this encounter
could have occurred only during the transfer phase.

\begin{lemma}[Disturbed recovery]\label{lem:disturbed}
    Assume that the history has been admissible up to a time when
    the head steps on an island $J_{1}$,
    in a transfer sweep $\TransferStart(\delta)$, $\delta\in\{-1,1\}$
    or in the first zigging into the neighbor colony immediately after this sweep.

    Then the annotation can be extended such that
    a relief event of duration $\le 3\K_\R $ and extent $\le 3\E$ occurs.
\end{lemma}
\begin{Proof}
    Without loss of generality,
assume $\delta = 1$, that is the direction of the transferring sweep
    is to the right.
    Let $\chi$ denote the healthy configuration that is part of the annotation at
    the time when an alarm occurs at some cell
 \begin{align*}
   x_{0}.
 \end{align*}
    In what follows, we will sometimes refer to $\chi.\front$ of this moment as just
    the \df{front}.

    In the transferring phase all structural (that is $\cCore$)
    information  in the non-island cells we pass is
    computable from the field $\Core$ of the state (see Lemma~\ref{lem:coord}).
    Therefore if no alarm or burst occurs while the head is
    on $J_{1}$, then the part of $J_{1}$ that was passed can be
    deleted from the island.

    If no burst occurs within the next $2\K_\R$ steps, then
    Lemma~\ref{lem:alarm} is applicable.
    From now on, we assume that a burst occurs during this time,
    creating an island $J_{0}$.

    If the burst occurs while the head is on island $J_{1}$,
    then Lemma~\ref{lem:burst} is applicable.
    Assume therefore that alarm occurs at some time
    while the head is on $J_{1}$ (or over a cell next to it),
    but a burst occurs only at some later time $t_{1}$.
    Let $\D(t)$ denote the interval of marked cells at time $t$,
    created by the recovery process started by the alarm.

    \begin{step+}{step:disturbed.alarm-soon}
        Suppose that new alarm will be called within
        $2\Z$ steps after the burst at some cell $x$.
    \end{step+}
    \begin{prooofi}
      Then zigging implies that
      we are also not removed beyond distance $\Z$ from $\D(t_{1})$ at the
      time of the alarm.

        After the burst, we are within distance $\beta$ from $\D(t_{1})$.
        If the recovery before the burst did not determine $z_{1}$ yet, then the
        size of $D(t_{1})$ is at most $7\beta$.
        Since after the burst we are burst-free for a while,
        Lemma~\ref{lem:alarm} guarantees the relief.
      Otherwise, recovery after the initial alarm has already defined cell
      $z_1$ of the recovery procedure.
        Then,
         \begin{align}\label{eq:alarminterval}
             x_{0} \in \lint{z_{1} -0.3\E}{z_{1} + 0.2\E}.
         \end{align}
        A new recovery area $\R' = z'_1 + [-\E,\E)$ will be created.
        The alignment guarantees $z'_1 = z_1$, $z_1 - \E$ or $z_1 + \E$.
        The direction $\delta$ computed after the second alarm is
        necessarily the same as the one computed after the first one.

        Now, if alarm after the burst is called in
        the same interval~\eqref{eq:alarminterval} as the initial alarm
        then the same recovery
        interval will be opened again, hence $z'_1 = z_1$.
        If $x < z_1 -0.3\E$, then $z'_1 = z_1 - \E$.
        Finally, if alarm after the burst is called to the right
        of $z_1 + 0.2\E$, then $z'_1 = z_1 + \E$.

        If $z'_1 = z_1$, then all cells of $\D(t_1)$
        will be reprocessed, and
        the recovery succeeds.
    \begin{step+}{step:disturbed.alarm-soon.left}
        Assume $z'_1 = z_1 - \E$.
    \end{step+}
    \begin{prooofi}
        If $\chi.\front < z'_1 + 0.3\E$ then $\widehat{\R'} = \R'$.
        Then after the new recovery finishes,
        marked cells in interval $\D(t_{1})\setminus \R'$ of length $\leq \E$
        may still be there.
        However, the mode after the recovery is normal, and
        we have assumed that the sweep direction is to the right.
        Therefore, these marked cells will be reached, and
        alarm will be triggered.
        Indeed, even if the front is at the colony boundary,
        and $z_{1}$ is the colony boundary (in which case the head is turning
        left), within $\Z - 4\beta$ steps
        the zigging will start, and the head will pass over
        $z_{1}$, where marked cells may exist.
        If they exist they trigger alarm, and an undisturbed recovery,
        with a recovery interval equal to $\R$, will eliminate
        remaining marks.
        \footnote{We allow the head to zig
        into the neighbor colony in order to definitely reach all remaining
        marked cells.}

        If $\chi.\front \ge z'_{1} + 0.3\E = z_{1} - 0.2\E$,
        then $\widehat{\R'} = (\R')_0^4$.
        Once the recovery over $\R'$ finishes, the head will be left on its
        right end, where alarm will be called, since marked cells will be found.
        Then, a new undisturbed recovery cleans the remaining
        marks and in the previous case.
      \end{prooofi}

      \begin{step+}{step:disturbed.alarm-soon.right}
        Consider the case $z'_{1} = z_{1}+\E$.
      \end{step+}
      \begin{prooofi}
          Now the new recovery interval does not contain the front $0.2\E$ deep
          inside.
          Indeed, alarm at $x_{0}$ was called on an island
          either when moving right, or while zigging into a right neighor
          colony (this zigging goes at most $4\beta$ deep).
          Therefore,  the position where alarm is called for the first time
          can only be to the right of the front
          within a distance not exceeding $\Z  - 4\beta$
          (see Fig.~\ref{fig:disturbed-recovery1}).

         \begin{figure}[!ht]
            \centering
              \includegraphics[height=0.7in]{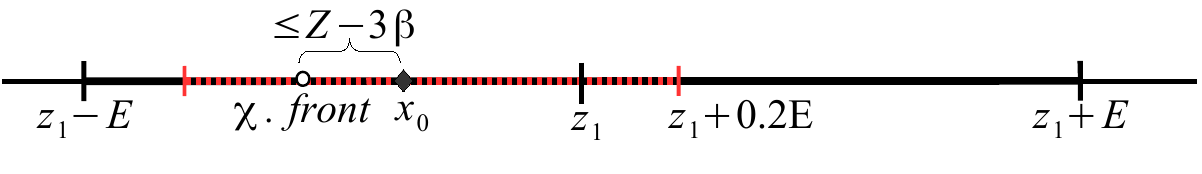}
                \caption{Point
                $x_0$ where the alarm is called once the head encountered the island,
                is always to the left of $z_{1} + 0.2\E$, therefore
                $\front(\chi) < z_{1} + 0.2\E$ as well \label{fig:disturbed-recovery1}
            }
        \end{figure}

          Since $x_0 \in \lint{z_1 -0.3\E}{z_1 + 0.2\E}$, the front cannot be in
          $\lint{z'_{1} - 0.3\E}{z'_{1} +0.3\E}$, and the Correction
          Lemma~\ref{lem:correction} yields $\widehat{\R'} = (\R')_4^0$.

          Once the recovery completes, the head is put
          into $z_1$, where a new alarm will be called when the marked cells are
          discovered during zigging.
          The new recovery area after this alarm is $\R$ again, and
          the process eliminates the remaining parts of $\D(t_1)$,
          leading to relief.

%
%

        \end{prooofi} 
    \end{prooofi} 

    \begin{step+}{step:disturbed.alarm-later}
        Suppose that alarm will not be called within $2\Z$ steps after the burst.
    \end{step+}
    \begin{prooofi}
        \begin{step+}{step:disturbed.normal}
            Suppose that the burst brings the machine to normal mode.
        \end{step+}
        \begin{prooofi}
            If $J_{0}\supseteq \D(t_{1})$ then the proof of
            Lemma~\ref{lem:burst} is applicable.
            Otherwise, as zigging meets the marked cells in
            $\D(t)$ within $2\Z$ steps, a new alarm will be called,
            and part~\ref{step:disturbed.alarm-soon} is applicable.
        \end{prooofi}

        If a burst occurs while the head is near the boundary of the recovery
        interval, then it may leave an island outside the recovery
        interval (within distance of $\E + \beta$ from $z_1$),
        provided that after the burst the recovery continues seamlessly where it
        was interrupted.

    \begin{step+}{step:disturbed.to-marking}
        Suppose that the stage after the burst is $\Marking$.
    \end{step+}
    \begin{prooofi}
        If the recovery process continues the old one seamlessly, then it
        terminates with success.

        Otherwise, since the marking stage employs zigging,
        alarm occurs within $2\Rec.\Z < 2\Z$ steps.
        From then on, an analysis identical to the
        one in the proof of Lemma~\ref{lem:burst} shows
        that the cells marked after the burst
        will be contained in the recovery area created by the new alarm.
        To what happens after, the analysis of
        part~\ref{step:disturbed.alarm-soon} is applicable.
    \end{prooofi} 

    \begin{step+}{step:disturbed.to-post-marking}
        Suppose that the stage after the burst is $\Planning_{i}$ or $\Mopping$.
    \end{step+}
    \begin{prooofi}
        Since these stages expect to walk over a recovery area,
        they must seamlessly continue what went on before, except
        for changing the state and the content of some cells in an
        island---otherwise alarm occurs immediately.

        If the burst occurs during $\Planning_{1}$,
        and it changes what is computed, then
        $\Planning_{2}$ will notice this and trigger alarm.
        Since this alarm occurs in the existing marked
        area $\D(t_{1})$, the analysis of
        part~\ref{step:disturbed.alarm-soon} still applies.

        If the burst occurs during $\Planning_{2}$
        or $\Mopping$ then it either triggers alarm,
        in which case the above analysis applies,
        or it allows the recovery process to end, with the lasting effect of the
        burst restricted just to the island $J_{0}$.

        Lemma~\ref{lem:correction} guarantees
        that whatever assignments $\cCore\gets\cRec\Core$
        were made in the mopping stage, they are admissible;
        even if mopping will be interrupted by a burst
        (and then continued as mopping).
    \end{prooofi} 
    \end{prooofi} 

    To bound the duration of relief, we note that at the worst
    case in part~\ref{step:disturbed.alarm-soon},
    the first recovery initiated by the island can reach only up to mopping.
    After the burst, at most two other full
    recovery cycles occur with at most $2\Z \ll \E$
    steps before them.
    Hence the total duration of the relief is
    $\leq 3\K_\R$.

\end{Proof}

\subsection{Finishing the proof}

The following lemma implies the main theorem.

\begin{lemma}[Recovery]\label{lem:recovery}
    Assume that machine $M_{1}$ starts working on a tape configuration
    of the form $\varphi_{*}(x)$.
    Every \df{$(\beta,V)$-noisy trajectory} of $M_{1}$ can be annotated.
\end{lemma}

\begin{Proof}
    Assume that the history has been annotated in an
    admissible way up to a certain time.
    First we show that in case a distress event occurs,
    the annotation can be extended to keep
    property~\ref{def:annotated-hist}
    (\ref{i:annotated-hist.distress}).
    Then using this, we will show that in case of no distress,
    the annotation can also be extended in an admissible
    way while keeping the other  properties.

    \begin{step+}{step:recovery.distress}
        Consider property~\ref{def:annotated-hist}
        (\ref{i:annotated-hist.distress}).
    \end{step+}
    \begin{prooofi}
        If a distress event occurs due to a burst then
        Lemma~\ref{lem:burst} applies.

        Assume now that a distress event occurs due to
        stepping onto an island $J_{1}$.

        Assume first that no burst occurs in the following $3\K_{\R}$ steps.
        Now, if no alarm sounds within $2\Z$ steps,
        then zigging guarantees that the part of the island passed over can be
        replaced with a stain.
        (The only way not to pass some part is when the island is in a neighbor
        colony at distance $\approx 4\beta$ from the boundary: zigging may reach
        just a part of it.)
        If alarm sounds, then Lemma~\ref{lem:alarm} is applicable.

        If there will also be a burst within the
        following $3\K_\R$ steps, then there has not occurred
        any burst recently (within $V$ steps).
        There could not have been any distress in the last sweep:
        indeed, any earlier island on which the head could have
        stepped would have been eliminated (at least its
        part in the path of the head) without or with alarm,
        as seen in the previous paragraph.
        But then the only way to step on an island is under
        the conditions of Lemma~\ref{lem:disturbed}.
    \end{prooofi} 

    \begin{step+}{step:recovery.progress}
        Consider property~\ref{def:annotated-hist}
        (\ref{i:annotated-hist.progress}).
    \end{step+}
    \begin{prooofi}
        Assume an admissible annotated history until
        a distress-free time $t$.
        We will show that by just keeping the islands
        constant, the annotation is extendable in an
        admissible way to $t+1$.
        In particular, there will still be a satisfying
        healthy configuration.

        Looking at Definition~\ref{def:healthy} of
        healthy configurations, most properties are
        obviously preserved in each step by just the
        form of the transition rule.
        The exceptions are the property
        which requires that the $\cDrift$ track holds constant
        values in certain intervals at certain times,
        and the property which requires that
        $\cInfo$ and $\cState$ tracks hold
        valid codewords of the code $(\varphi_{*}, \varphi^{*})$
        defined in Section~\ref{subsec:coding}.

        So we are only concerned with the recomputation of the values
        of $\cState$, $\cInfo$, $\cDrift$ in the base colony,
        during the computation phase, and then the
        transfer of $\cState$ during the transfer phase.
        (The value of $\cDrift$ in the neighbor colonies
        is inherited from earlier,
        and its spreading from $\Drift$ is watched over by
        the coordination requirement: a change would
        trigger alarm at zigging.)

        Recall the structure and the tasks of the
        computation phase given in~\ref{subsec:computation-phase}.
        By properties of annotated configurations in
        Definition~\ref{def:annotated-config},
        in the base colony, besides a possible island,
        there is at most 1 more stain of size $\beta$,
        and possibly more stains, all contained in a single interval
        of size $\E + \beta$.
        (The bound comes from the length of possible penetration of the head
        in a neighboring colony while faults could occur.)
        These last stains can be ignored, since our code is defined in such a
        way that it places a codeword of the $(\beta,2)$ burst-error-correcting
        code $(\upsilon_{*},\upsilon^{*})$ at a distance $1.1\E$ away from the
        colony boundaries.

        The recovery rules do not change the
        $\cHold$, $\cState$ and $\Info$ tracks, and
        given that $\Sweep<\TransferStart$,
        they do not change $\cDrift$ track either.
        Therefore, since there are at most 2 stains at distance $1.1\E$ from the
        boundaries, and our code is $(\beta,2)$ burst-error-correcting,
        the result of decoding from the $\Info$ and $\State$
        tracks during the computation phase
        will be the same as if the configurations had been stainless all along.

        Even if a fault causes the head to step into
        a neighbor colony that can be empty and set
        $\Sweep = \TransferSw(\pm 1)$,
        after at most $2\Z$ steps, the head will step back inside
        the colony it came from, and it will call alarm there.
        Since $\E \gg \Z$, the distress area will contain entirely
        this segment of cells.

        Any distress event will directly affect at most one
        of the three repetitions of the computation phase:
        the configuration is centrally consistent during the others.
        Consequently, the correct values will will be
        stored in track $\cHold[i]$, $i\in\{1,2,3\}$ for all but one $i$.
        If the sweep of the field majority computation
        during the encoding stage of the computation phase
        is distress-free, then every  cell will receive
        the correct value $\maj(\cHold[1\dots 3])$.
        But even if distress occurs in this sweep, relief
        guarantees that all cells but the ones in the
        island of the burst that caused the distress will hold the
        correct value.

        The same argument proves the property that the
        newly computed $\cState$ will be correctly transferred
        to the extended base colony in the transfer phase.
    \end{prooofi} 

    \begin{step+}{step:recovery.stains}
        Consider property~\ref{def:annotated-hist}\eqref{i:annotated-hist.stain}.
    \end{step+}
    \begin{prooofi}
        From the above argument it is clear that
        the only possible stain remaining in the
        base colony is the one created by a burst
        in the current work period.
        On the other hand, we can add stains and
        islands to neighbor colonies.

        Let us see what is the farthest distance to which we can intrude
        into a neighbor colony and leave islands.
        With zigging, the head can penetrate
        at most $4\beta$ cells into the neighbor colony,
        where it can find an island causing alarm.
        A burst ocurring anywhere in the recovery interval created by
        this alarm may leave a stain anywhere
        within distance of $\E + \beta $ from the colony boundary (where
        the recovery interval is centered).
    \end{prooofi} 
\end{Proof}

\begin{lemma}[Simulation]\label{lem:simulation}
    Under the conditions of Lemma~\ref{lem:recovery},
    via some simulation function $\Phi^{*}$
    (to be defined in the proof of the  present lemma),
    the movement of the base colony corresponds to the
    head movement of the simulated machine $M_{2}$
    (scaled up by a factor of $Q$).
    Whenever the sweep in the free cells of the base colony
    is not one of switching to a new work period,
    the array of $\cState$ values there decodes into the state of $M_{2}$, and
    the array of $\cInfo$ values decodes
    into the current tape cell symbol of $M_{2}$.
\end{lemma}
\begin{proof}
    Lemma~\ref{lem:recovery} gives us an admissible history.
    At all distress-free times, it also defines uniquely a
    base colony.
    For distressed times, let the base colony be equal to that
    of the last distress-free time.
    Once a base colony is given for each configuration, the
    simulation function is also uniquely defined: we decode
    the simulated cell content of each cell of
    $M_{2}$ from the corresponding colony, and the simulated
    state from the $\cState$ array of the base colony.
    Part~\ref{step:recovery.progress} of the proof of
    Lemma~\ref{lem:recovery} shows that the
    decoding indeed defines a trajectory of $M_{2}$.
\end{proof}

\begin{proof}[Proof of Theorem~\protect\ref{thm:main}]
    The statement follows essentially from Lemma~\ref{lem:simulation},
    adding only the following.
    Let $f$ be a projection from the alphabet of $M_1$ to the alphabet of $M_2$,
    defined by $f(s)=s.\cInfo$.
    Consider now the cell at the origin of the tape.
    Then, relation~\eqref{eq:main-thm}
    holds due the step~\ref{i:comp.write} of the computation procedure in
    section~\ref{subsec:computation-phase}.

    What are all the lower bounds on $Q$?
    Since the program of the machine $M_2$ must fit in a colony,
    $Q$ is lowerbounded by $p_2$.
    Definitions~\eqref{eq:Z}  and~\eqref{eq:Expansion} show $\E=O(\beta)$.
    We needed to be able to define the code $(\phi_{*},\phi^{*})$ in
    Section~\ref{subsec:coding} fitting into the part of the colony away by
    $1.1\E$ from the boundary.
    These requirements are satisfied with $Q$ depending linearly on
    $\log{|\Sigma_2|}$, $\log{|\Gamma_2|}$ and $\E$.

    The computation phase lasts $O(Q)$ steps, where we also used the requirement
    in Definition~\ref{def:univ-TM}.
    The transferring of the $\State$ into the neighboring
    colony will need $Q$ sweeps, that is $O(Q^{2})$ steps.
    Therefore the constant $V$ bounding the time
    overhead of machine $M_1$ is $V = O(Q^2)$.
\end{proof}

\section{Conclusions and future work}

    In this paper we have shown that for any Turing machine there is one that
    can simulate it while while correcting occasional violations
    of its own transition function.
    The procedure recovering the simulation structure is based on an
    organization in which any
    group of cells affected by the faults is
    surrounded by cells that conserve some valid traces
    of the computation.

    We hope to use this construction, similarly to~\cite{GacsSorg97},
    as a building block in a more complex construction of a Turing
    machine that can resist faults occurring independently with
    small probability.

    To the best of our knowledge, this is the first construction
    of a reliable sequential machine.
    An interesting question is if the Turing machines are the
    simplest machines that can perform universal computation
    under isolated bursts of noise.
    It seems that simpler models, like the counter machines of~\cite{Minsky},
    are insufficient, but
    there are some interesting questions open concerning the nature of their
    insufficiency.

\bibliographystyle{amsplain}

\end{document}